\documentclass{siamart171218}

\usepackage[utf8]{inputenc}
\usepackage{amsmath}
\usepackage{amsfonts}
\usepackage{amssymb}
\usepackage{mathtools}
\usepackage{subcaption}
\usepackage{appendix}
\usepackage[shortlabels]{enumitem}
\usepackage{braket}

\usepackage{setspace}
\allowdisplaybreaks

\numberwithin{equation}{section}

\newtheorem{remark}{Remark}[section]

\renewcommand{\tilde}{\widetilde}

\renewcommand{\bar}{\overline}

\newcommand*{\dd}{\mathop{}\!\mathrm{d}}

\newcommand{\trans}{\mathsf{T}}
\newcommand{\hermi}{\mathsf{H}}
\DeclareMathOperator{\ran}{Ran}
\DeclareMathOperator{\spn}{Span}

\usepackage{xcolor}

\numberwithin{figure}{section}

\newcommand{\Hbulk}{H_{\text{bulk}}}
\newcommand{\Hledge}{H^{\text{L}}_{\text{edge}}}
\newcommand{\Hedge}{H_{\text{edge}}}
\newcommand{\Hcedge}{H^{\text{cont}}_{\text{edge}}}
\newcommand{\Hredge}{H^{\text{R}}_{\text{edge}}}
\newcommand{\Hledged}{{H}^{\text{L}}_{\text{edge}}}
\newcommand{\Hledget}{{H}^{\text{L}}_{\text{trunc}}}

\newcommand{\Hzzedge}{H^{\text{ZZ}}_{\text{edge}}}
\newcommand{\Hzzedged}{{H}^{\text{ZZ}}_{\text{edge}}}

\newcommand{\Gammatrunc}{\Gamma_{\text{trunc}}}
\newcommand{\HGammatrunc}{(\Hledged \Gamma)_{\text{trunc}}}
\newcommand{\psijtrunc}{\psi_{j,\text{trunc}}}
\newcommand{\per}{\text{per}}
\newcommand{\fin}{\text{fin}}
\DeclareMathOperator{\nul}{Nul}
\newcommand{\dash}{-{}- }

\newtheorem{assumption}{Assumption}[section]

\newcommand{\ip}[2]{\ensuremath{ \left< \left. #1 \right| #2 \right> } }

\usepackage{subcaption}
\usepackage{soul}
\usepackage[backend=bibtex,style=numeric,citestyle=numeric-comp,maxbibnames=99]{biblatex}
\renewbibmacro{in:}{\ifentrytype{article}{}{\printtext{\bibstring{in}\intitlepunct}}}
\usepackage{ragged2e}

\bibliography{bibliography.bib}

\title{Computing edge states without hard truncation}

\author{Kyle Thicke, Alexander B. Watson, and Jianfeng Lu}

\begin{document} 

%TO DO AW:
%\input{todo_AW}

%\newpage

\maketitle

\begin{abstract}
We present a numerical method which accurately computes the discrete spectrum and associated bound states of Hamiltonians which model electronic ``edge'' states localized at boundaries of one and two-dimensional crystalline materials. The problem is non-trivial since arbitrarily large finite ``hard'' truncations of the Hamiltonian in the infinite bulk direction tend to produce spurious bound states partially supported at the truncation. Our method, which overcomes this difficulty, is to compute the Green's function of the Hamiltonian by imposing an appropriate boundary condition in the bulk direction; then, the spectral data is recovered via Riesz projection. We demonstrate our method's effectiveness by studies of edge states at a graphene zig-zag edge in the presence of defects modeled both by a discrete tight-binding model and a continuum PDE model under finite difference discretization. Our method may also be used to study states localized at domain wall-type edges in one and two-dimensional materials where the edge Hamiltonian is infinite in both directions; we demonstrate this for the case of a tight-binding model of distinct honeycomb structures joined along a zig-zag edge. 
\end{abstract}

\section{Introduction} \label{sec:introduction} \setcounter{page}{1}

When computing the electronic states of crystalline materials it is typical to first treat the material as infinite in all directions of periodicity, leading to an periodic bulk Hamiltonian which is naturally analyzed using the Bloch transform (see \cite{ashcroft_mermin}, for example). The eigenstates of this Hamiltonian, known as bulk states, are quasi-periodic with respect to the crystal lattice and hence extensive throughout the material. Materials may also host electronic states which decay rapidly away from the physical edge of the material known as edge states. Edge states are bound states of the semi-infinite Hamiltonian (known as the edge Hamiltonian) obtained by truncating the infinite bulk Hamiltonian in one direction and imposing a Dirichlet boundary condition at the truncation. Such states play an important role in the theory of the quantum Hall effect and of topological insulators, and are of independent interest for wave-guiding applications because of their robustness to certain classes of local perturbations \cite{1981Halperin,1993Hatsugai,2014LuJoannopoulosSoljacic,2018Bal,2018Ozawaetal}.

In this work we propose a numerical method for computing the discrete spectrum and associated bound states of such semi-infinite edge Hamiltonians. The problem is non-trivial by the fact that arbitrarily large finite truncations of the Hamiltonian produce spurious bound states partially supported at the truncation which do not correspond to bound states of the original Hamiltonian. The key idea of our method, which overcomes this difficulty, is to instead compute the Green's function of the edge Hamiltonian along an appropriate contour by truncating the problem while imposing an appropriate boundary condition at the truncation. The spectrum and eigenstates of the original Hamiltonian may then be recovered using the Riesz projection formula. Our method relies on eventual periodicity of the Hamiltonian which implies that exact boundary conditions for computing the Green's function can be derived from a transfer matrix; similar ideas have been explored for numerical computation of local defects in electronic structures in \cite{2016LiLinLu}, and for analytical calculations e.g.,  \cite{2016DwivediChua,2017PengBaovonOppen}. We provide self-contained proofs of the existence of exact boundary conditions for computing the Green's function for arbitrary finite range discrete Hamiltonians and that all spectral data can be recovered exactly from the Green's function. Hence our method is exact up to rounding errors for discrete models and exact up to discretization error of continuum models. 

The rest of the paper is organized as follows. In Section \ref{sec:SSH_pollution} we motivate our method by demonstrating the failure of a na\"ive method using a one-dimensional SSH-type model which has the essential features of the general case. In Section \ref{sec:GreenFunctionMethod} we introduce our method and show that it accurately computes spectral data of this model. In Section \ref{sec:trans_mat_in_general} we generalize our method to a class of models whose Hamiltonians are ``eventually periodic'', and prove the existence of exact boundary conditions for computing the Green's function: first under a mild simplifying assumption on the Hamiltonian (Section \ref{sec:BCs_in_general}) and then in generality for finite range Hamiltonians (Appendices \ref{sec:method_in_general} and \ref{sec:method_in_general_proofs}). In Section \ref{sec:appToHoneycomb} we present results of computations using our method of states localized at edges of two-dimensional honeycomb structures modeled both in the tight-binding limit and by a continuum Schr\"odinger equation under finite difference discretization.

\medskip

\noindent {\bf Acknowledgements:}  This work is supported in part by National Science Foundation through grants DMS-1454939 and ACI-1450280 and Department of Energy through grant DE-SC0019449. The authors would also like to thank Jacob Shapiro and Michael I. Weinstein for stimulating discussions. Early stages of this project were carried out by Qinyi Zhu as an undergraduate researcher visiting Duke University during the Summer of 2017.

\section{Spurious bound states in a one-dimensional model: the SSH model} \label{sec:SSH_pollution}

In this section, in order to motivate and clearly lay out the ideas behind our method, we restrict attention to a discrete one-dimensional model which has the essential features of the general case: the Su-Schrieffer-Heeger (SSH) model with real, nearest-neighbor hopping \cite{1979SuSchriefferHeeger}. The structure of this section is as follows.

In Section \ref{sec:bulk_SSH}, we review the definition and spectrum (known as the bulk spectrum) of the SSH bulk Hamiltonian. In Section \ref{sec:SSH_with_edge}, we introduce an SSH edge Hamiltonian with defects whose essential spectrum is equal to that of the bulk Hamiltonian but whose discrete spectrum and associated bound states are difficult to compute by hand. In Section \ref{sec:hard_trunc} we then demonstrate the appearance of spurious bound states in na\"ive numerical computations. 

The stage will then be set for us to present in Section \ref{sec:GreenFunctionMethod} our main method which we refer to as the Green's function method. This method accurately computes the discrete spectrum and associated bound states of the SSH edge Hamiltonian with defects within the gap between essential spectrum.

\subsection{SSH bulk Hamiltonian and bulk spectrum} \label{sec:bulk_SSH} 

In this section, we review the definition and spectrum of the SSH bulk Hamiltonian. We consider an electron (we ignore spin) ``hopping'' along a one-dimensional infinite periodic lattice (chain) with hopping amplitudes that alternate between sites along the chain.  A fundamental cell of the chain consists of any two neighboring sites, which we label $A$ and $B$. We introduce the notation
\begin{equation}
	\psi_m = \begin{bmatrix} \psi_m^A \\ \psi_m^B \end{bmatrix} \in \mathbb{C}^2, \quad m \in \mathbb{Z},
\end{equation}
to represent the restriction of the electron wavefunction $\psi \in l^2(\mathbb{Z};\mathbb{C}^2)$ to the $m$th fundamental cell of the chain. Denoting the intra- and inter-cell hopping amplitudes by $t_1$ and $t_2$ (for simplicity, we ignore all hopping other than between nearest-neighbors and assume that $t_1$ and $t_2$ are real and non-zero), the bulk Hamiltonian $\Hbulk$ is 
\begin{equation} \label{eq:ssh_bulk}
\begin{split}
	[ \Hbulk \psi ]_m =  \phantom{.} &A^* \psi_{m-1} + V \psi_m + A \psi_{m+1} \quad m \in \mathbb{Z}, \\
        &V := \begin{bmatrix} 0 & t_1 \\ t_1 & 0 \end{bmatrix} \quad A := \begin{bmatrix} 0 & 0 \\ t_2 & 0 \end{bmatrix}.
\end{split}
\end{equation}
Periodicity of the bulk Hamiltonian \eqref{eq:ssh_bulk} implies, via Bloch's theorem (see \cite{ashcroft_mermin,reed_simon_4}, for example), that all (generalized) eigenfunctions $\Phi$ of $\Hbulk$ are $k$-quasi-periodic: $\Phi_{m+1}(k) = e^{i k} \Phi_m(k) \; m \in \mathbb{Z}$ for $k \in [-\pi,\pi]$. A standard calculation then shows that the associated eigenvalues of these eigenfunctions are given by
\begin{equation} \label{eq:epair_cell}
	E^\pm(k) = \pm \left| t_1 + e^{- i k} t_2 \right| \quad k \in [-\pi,\pi].
\end{equation}

The spectrum of \eqref{eq:ssh_bulk}, $\sigma(\Hbulk)$, is therefore the union of the real intervals (known as spectral bands) swept out by the functions \eqref{eq:epair_cell} as $k$ is varied over the interval $[-\pi,\pi]$. If $\frac{ |t_2| }{ |t_1| } = 1$, then $\sigma(\Hbulk)$ consists of a single interval. If $\frac{ |t_2| }{ |t_1| } \neq 1$, then $\sigma(\Hbulk)$ consists of two intervals, separated by a gap which is symmetric about $0$.

The (generalized) eigenvectors of $\Hbulk$ corresponding to the eigenvalues $E^\pm(k)$ are known as bulk states. These states do not decay as $m \to \pm \infty$.

\subsection{SSH edge Hamiltonian and edge spectrum} \label{sec:SSH_with_edge}

We define a SSH left edge Hamiltonian which acts on $l^2(\mathbb{N};\mathbb{C}^2)$ as follows
\begin{equation} \label{eq:ssh_edge_nontriv}
\begin{split}
	&\left[ \Hledge \psi \right]_m = A^*(m-1) \psi_{m-1} + V(m) \psi_m + A(m) \psi_{m+1}  \quad m \in \mathbb{N}  \\
        &V(m) := \begin{bmatrix} V^A(m) & t_1(m) \\ t_1(m) & V^B(m) \end{bmatrix} \quad A(m) := \begin{bmatrix} 0 & 0 \\ t_2(m) & 0 \end{bmatrix}   \\
	&\psi_0 = 0.
\end{split}
\end{equation}
Here and in what follows we adopt the convention that the natural numbers start at $1$, i.e. $\mathbb{N} = \left\{1,2,...\right\}$. Note that we allow the edge to have defects, modeled by (non-zero) hopping amplitudes whose values depend on the cell index $m$, and for possibly non-zero real onsite potentials $V^\sigma(m), \sigma\in\{A,B\}$. We make the following assumption on the behavior of these functions for large $m$. 
\begin{assumption}[Periodicity of bulk medium] \label{as:bulk_periodicity}
There exists a positive integer $M \geq 1$ such that for all $m \geq M$:
\begin{itemize}
\item $t_1(m) = t^\infty_1$ and $t_2(m) = t^\infty_2$, where $t^\infty_1$ and $t^\infty_2$ are non-zero real numbers
\item $V^\sigma(m) = 0$ for $\sigma \in \{A,B\}$.
\end{itemize}
\end{assumption}
\begin{figure}
\centerline{\includegraphics[trim=0mm 6mm 0mm 5mm, clip, scale=.7]{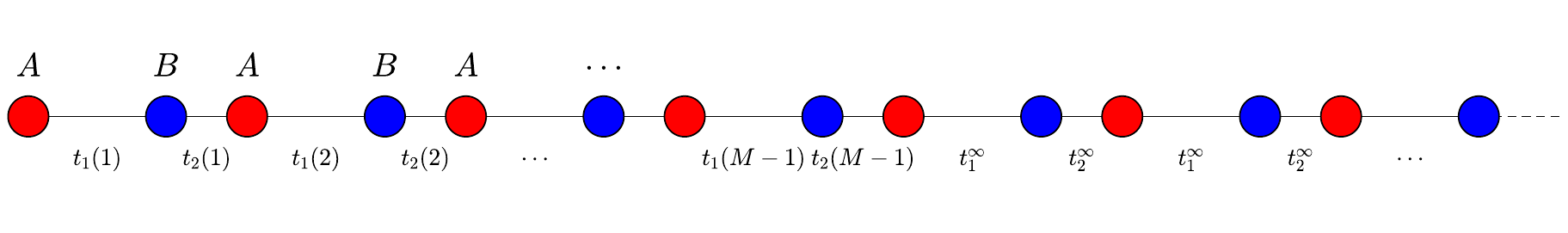}}
\caption{Schematic illustrating the SSH model, in particular Assumption \ref{as:bulk_periodicity}. $A$ and $B$ sites are colored red and blue respectively.  The hopping amplitudes between sites are labeled underneath the corresponding lines connecting the dots.  Assumption \ref{as:bulk_periodicity} states that for $m$ greater than or equal to some positive integer $M$ the hopping amplitudes $t_1(m)$ and $t_2(m)$ are constant and all onsite potentials $V^A(m), V^B(m)$ are zero. %\jl{how is "on-site" potential reflected in the figure?} \aw{It isn't. I'm not sure what would be a good way to do that in fact.}
} 
\label{fig:SSH_schematic}
\end{figure}
For a schematic illustrating Assumption \ref{as:bulk_periodicity}, see Figure \ref{fig:SSH_schematic}.

Under Assumption \ref{as:bulk_periodicity}, for $m \geq M + 1$ the Hamiltonians $\Hbulk$ and $\Hledge$ are identical and hence by the Weyl criterion
\begin{equation}
    \sigma_{ess}(\Hledge) = \sigma(\Hbulk),
\end{equation}
where $\Hbulk$ is defined by \eqref{eq:ssh_bulk} with $t_1$ and $t_2$ replaced by $t_1^\infty$ and $t_2^\infty$.

Regarding the discrete spectrum of \eqref{eq:ssh_edge_nontriv} we have the following standard result on existence of bound states with eigenvalue $E = 0$. %\aw{slight re-write here}
\begin{proposition} \label{prop:dis_zero_mode}
Assume that $V^\sigma(m) = 0$ for $\sigma \in \{A,B\}$ and all $m \in \mathbb{N}$. Then, if $\frac{|t^\infty_1|}{|t^\infty_2|} < 1$, then $E = 0$ is an eigenvalue of $\Hledge$, with associated bound state
\begin{equation} \label{eq:dis_lbstate}
	\psi_m = c \begin{bmatrix} \left( - \frac{ t_1(m-1) }{ t_2(m-1) } \right)^{m-1} \\ 0 \end{bmatrix} \quad m \in \mathbb{N},
\end{equation}
where $c$ is an arbitrary complex constant. If instead $\frac{|t^\infty_1|}{|t^\infty_2|} \geq 1$, $E = 0$ is not an eigenvalue of \emph{$\Hledge$}. 
\end{proposition}
The bound state \eqref{eq:dis_lbstate} is known as an edge state because it decays away from the edge of the material. %\jl{edge of the material?} \aw{yes, fixed now}
The existence of edge states can be related to topological invariants associated with the associated bulk Hamiltonian in a link known as the bulk-boundary correspondence (see \cite{1993Hatsugai,2013GrafPorta,Asboth,2011DelplaceUllmoMontambaux}, and references within, for example). We will adopt the convention that any bound state of the edge Hamiltonian is an edge state even if it is not associated with any bulk invariant.

An SSH right edge Hamiltonian $\Hredge$ can be defined on $l^2(-\mathbb{N};\mathbb{C}^2)$ analogously to $\Hledge$. Under an analogous assumption to Assumption \ref{as:bulk_periodicity} it is clear that $\sigma_{ess}(\Hredge)$ $= \sigma(\Hbulk)$ and that $\Hredge$ has a bound state with eigenvalue $0$ whenever $\frac{|t_2|}{|t_1|} > 1$ and $V^\sigma(m) = 0$ for $\sigma \in \{A,B\}$ and all $m \in - \mathbb{N}$. While \eqref{eq:dis_lbstate} is entirely supported on $A$ sites, the bound state of $\Hredge$ with eigenvalue $0$ is entirely supported on $B$ sites.

We emphasize that when $V^A(m)$ and $V^B(m)$ are non-zero and $t_1(m)$ and $t_2(m)$ are not constant, $\Hledged$ may have multiple eigenvalues in the gap between essential spectrum and $0$ may or may not be one of them (see Figure \ref{fig:SSH_exp}). In the next section we consider a na\"ive approach to computing all such eigenvalues of \eqref{eq:ssh_edge_nontriv} and their associated bound states.

\subsection{Spurious bound states associated with finite ``hard'' truncation of the semi-infinite Hamiltonian} \label{sec:hard_trunc}

Since bound states of \eqref{eq:ssh_edge_nontriv} must decay as $m \rightarrow \infty$, it is tempting to try to approximate such bound states and their associated eigenvalues by imposing a Dirichlet boundary condition at some large value of the cell index $m$, i.e. by computing solutions of the truncated eigenequation
\begin{equation}
	\Hledget \psi = E \psi \quad \psi \in l^2(\left\{ 1,...,\mathcal{M} \right\};\mathbb{C}^2)
\end{equation}
for some choice of $\mathcal{M} \gg M$, where $M$ is the integer appearing in Assumption \ref{as:bulk_periodicity} and where $\Hledget$ is defined by
\begin{equation} \label{eq:ssh_edge_nontriv_trunc}
\begin{split}
	&[ \Hledget \psi ]_m = A^*(m-1) \psi_{m-1} + V(m) \psi_m + A(m) \psi_{m+1} \quad m \in \{1,...,\mathcal{M}\}  \\
	&\psi_0 = \psi_{\mathcal{M}+1} = 0
\end{split}
\end{equation}
where $V(m)$ and $A(m)$ are as in \eqref{eq:ssh_edge_nontriv}. We call this method the hard truncation method. This approach produces inaccurate results because $\Hledget$ can have bound states partially supported at the truncation for arbitrarily large $\mathcal{M}$; see Figure \ref{fig:hardtrunc_edgestates}.

\begin{figure}
%plot_hardtrunc_edgestates(10,1,2)
\centerline{\includegraphics[scale=.75]{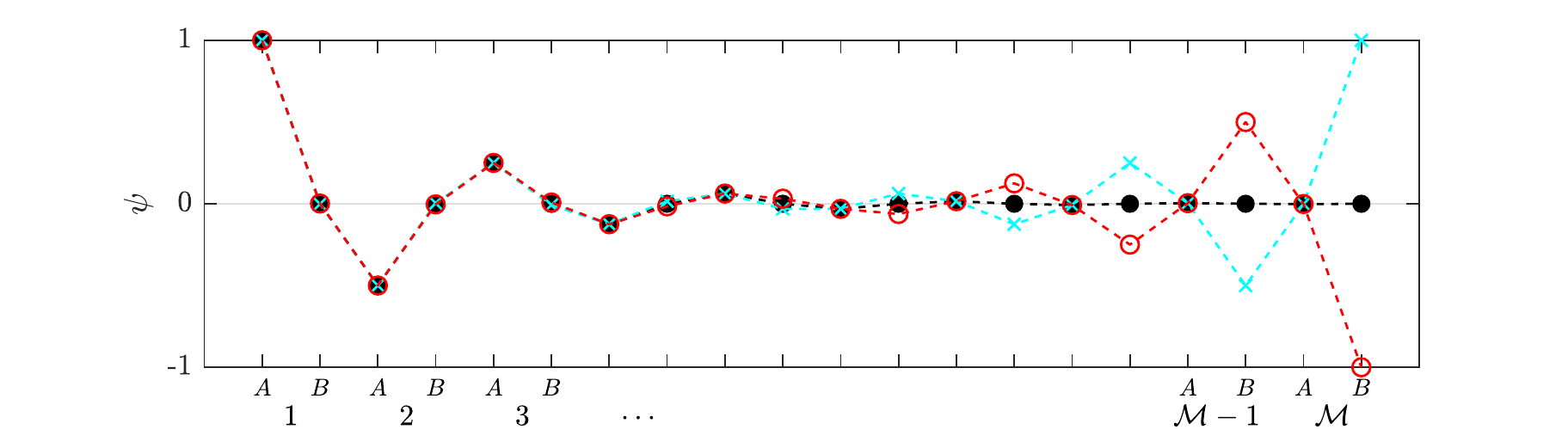}}
\caption{Bound states of the truncated left edge Hamiltonian $\Hledget$ \eqref{eq:ssh_edge_nontriv_trunc} (eigenvalues $\pm 9.16 \times 10^{-5}$, plotted with cyan $\times$s and red circles) when $t_j(m) = t_j^\infty, j \in \left\{1,2\right\}$ and $V^\sigma = 0, \sigma \in \{A,B\}$ for all $m$ and $\frac{|t_2^\infty|}{|t_1^\infty|} > 1$, compared with the unique bound state (eigenvalue $E = 0$) of the semi-infinite left edge Hamiltonian $\Hledge$ (black filled circles). The bound states arise from taking appropriate linear combinations of the zero energy edge states of the left and right edge Hamiltonians truncated to $l^2(\{1,...,\mathcal{M}\};\mathbb{C}^2)$. The situation is analogous to the case of two identical well-separated potential wells in the semi-classical regime \cite{1980Harrell,1984Simon} where one observes splitting (width $O(e^{-c \gamma})$ where $\gamma$ is the distance between wells and $c > 0$ is a constant) of the single potential well eigenvalues.} \label{fig:hardtrunc_edgestates}
\end{figure}

\begin{remark}
The problem described above can be avoided when $\frac{ |t_2^\infty| }{ |t_1^\infty| } > 1$ by truncating the chain in the middle of a cell (i.e. after an $A$ site) rather than at the end of a cell (i.e. after a $B$ site). The right edge is then well described by \emph{$\Hredge$}, the right-edge Hamiltonian, but with $t_1$ replaced by $t_2^\infty$ and $t_2$ replaced by $t_1^\infty$ which has no edge state in this case. The method we describe below has the advantage that it may be applied without modification to all possible choices of $t_1^\infty$ and $t_2^\infty$ equally well and also generalizes naturally to the more general cases we consider below.
\end{remark}

In the next section we introduce a method for computing spectral data of $\Hledged$ which overcomes this problem.

\section{Green's function method} 
\label{sec:GreenFunctionMethod}

In this section we introduce our main method which we refer to as the Green's function method.  Rather than trying to compute eigenvectors of $\Hledged$ directly we compute the associated spectral projection operator. This is done by computing the resolvent (Green's function) of the Hamiltonian and then recovering the projection via the Riesz projection formula. 
\begin{theorem}[Riesz projection] \label{th:Riesz}
Let $H$ be a self-adjoint operator on a Hilbert space $\mathcal{H}$. Let $\mathcal{C}$ denote a positively oriented contour in the complex plane which (1) does not intersect $\sigma(H)$, (2) encloses finitely many eigenvalues of $H$, and (3) does not enclose any other part of $\sigma(H)$. Then,
\begin{equation} \label{eq:Gam}
	\Gamma := \frac{1}{2 \pi i} \int_\mathcal{C} G(z) \dd z, \qquad G(z) := (z - H)^{-1} 
\end{equation}
is an orthogonal projection operator; in particular, it is the spectral projection onto the associated eigenspace of the eigenvalues enclosed in the contour $\mathcal{C}$.
\end{theorem}
Note that our method only yields information about eigenvalues within the contour $\mathcal{C}$. In practice this is not a problem: typically one is only interested in eigenvalues within the gap between essential spectrum (recall that for the SSH model $\sigma_{ess}(\Hledged)$ has a gap containing $0$ whenever $\frac{|t_2^\infty|}{|t_1^\infty|} \neq 1$). One therefore chooses the contour $\mathcal{C}$ in such a way as to encircle a large subinterval of the gap without intersecting with the essential spectrum.

In Section \ref{sec:GreenFunctionMethod_1} we explain how to compute the upper left block of the Green's function of $\Hledge$ for any $z \notin \sigma(\Hledge)$ and hence the upper left block of the spectral projection. In Section \ref{sec:getStatesFromGamma} we show how all relevant spectral information can be recovered using this method. In Figure \ref{fig:SSH_exp} we present results computed using our method.

\begin{figure} 
\centering
\begin{subfigure}[b]{.45\textwidth}
\includegraphics[scale=.41]{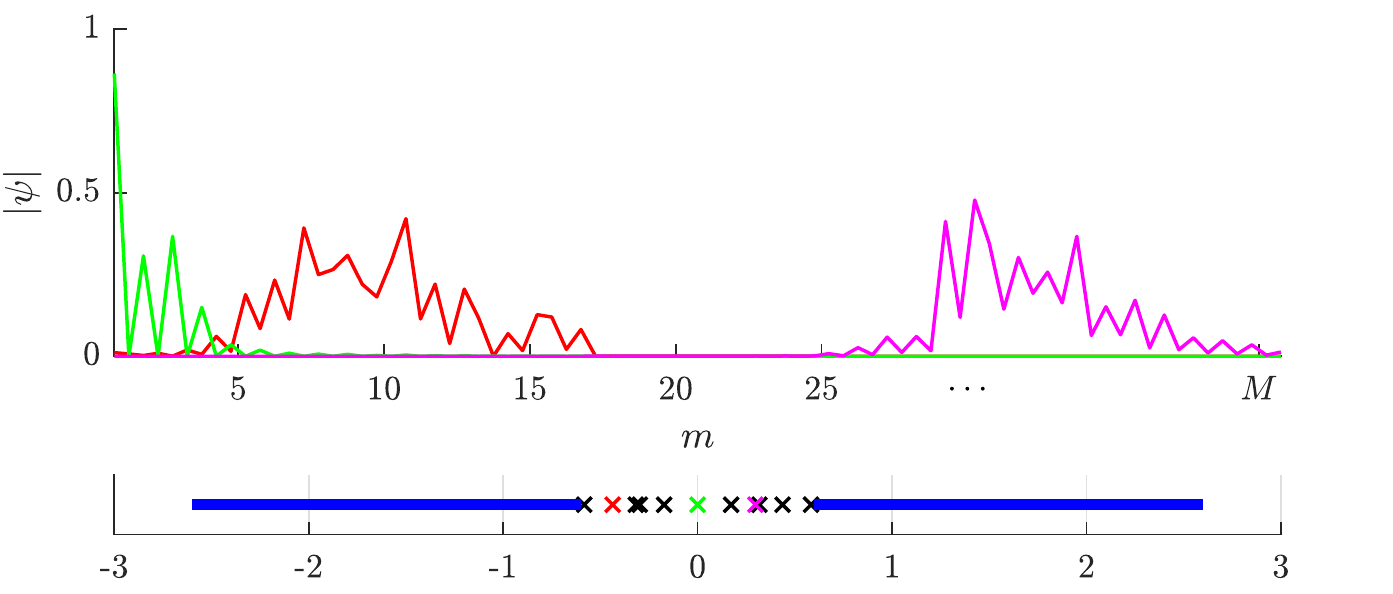}
\caption{\label{fig:SSH_exp_1}}
\end{subfigure}
\begin{subfigure}[b]{.45\textwidth}
\includegraphics[scale=.41]{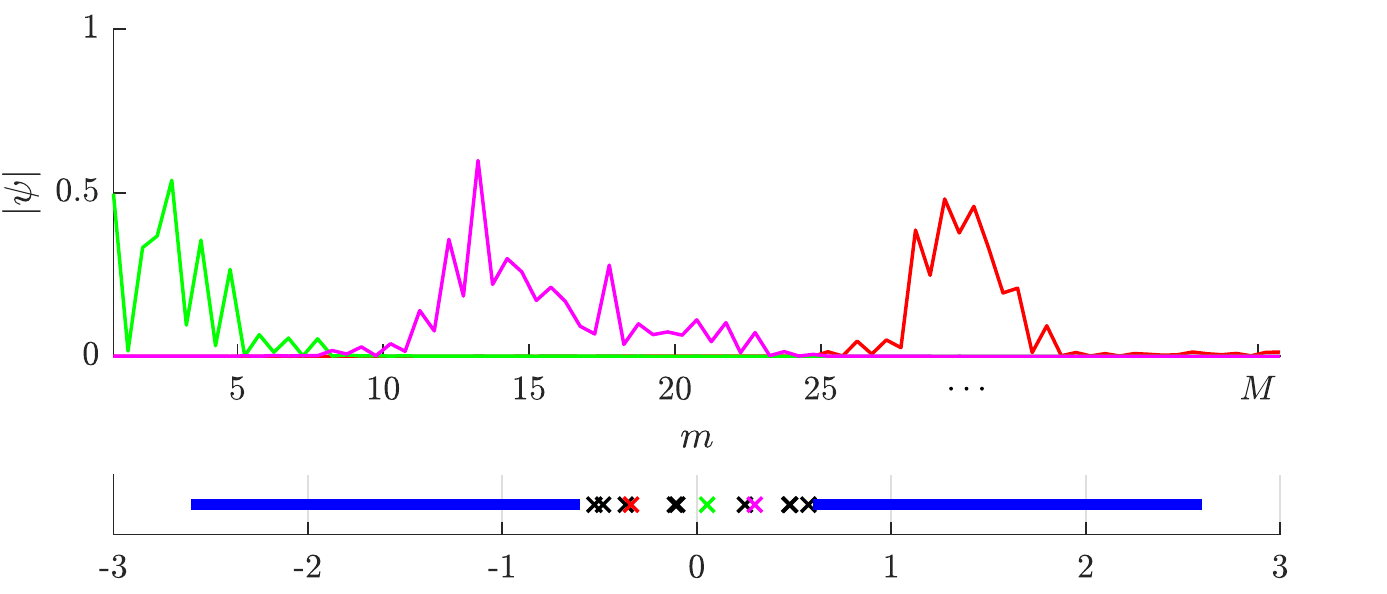} 
\caption{\label{fig:SSH_exp_2}}
\end{subfigure}
\caption{Example bound states (above, absolute value plotted) and spectra (below) of two realizations of the SSH edge Hamiltonian with defects \eqref{eq:ssh_edge_nontriv} computed with the Green's function method. In spectral plots, essential spectrum (computed using Bloch theory) is shown as a blue line, while eigenvalues are labeled by $\times$s. The associated eigenvalue of each bound state which is plotted in the above figures is plotted with the same color in the figures below. In both cases, we drew the values of the hopping amplitudes $t_1(m)$ and $t_2(m)$ for $m \leq 40$ from normal distributions with standard deviation $0.5$ and means equal to $t_1^\infty$ and $t_2^\infty$, with $t_1^\infty = 1$ and $t_2^\infty = 1.6$, respectively. To generate (a), we set all onsite potentials $V^\sigma(m) = 0$. In this case, we find that \eqref{eq:ssh_edge_nontriv} has an eigenvalue which is precisely $0$ (confirming Proposition \ref{prop:dis_zero_mode}). In addition, we see bound states with non-zero eigenvalues whose associated eigenfunctions are localized away from the edge. To generate (b), we drew the values of the onsite potentials $V^A(m)$ and $V^B(m)$ for $m < M$ from a normal distribution with standard deviation $0.5$ and mean $0$. In this case, we find again one bound state which decays rapidly away from the edge and multiple bound states whose associated eigenvectors are localized away from the edge. Note that the eigenvalue of the eigenvector which decays rapidly away from the edge is no longer $0$ in this case.}
%Using the Green's function method described in Sections \ref{sec:GreenFunctionMethod_1} and \ref{sec:getStatesFromGamma}, we computed two examples of discrete spectrum of the edge Hamiltonian \eqref{eq:ssh_edge_nontriv} lying in the gap between bands of essential spectrum when $t_2^\infty = 1.6$ and $t_1^\infty = 1$ with defects up to the $M = 40$th cell. Our results are shown in Figure \ref{fig:SSH_exp}.
\label{fig:SSH_exp} 
\end{figure}

%\begin{remark}
%Physically speaking, $\Gamma$ can be interpreted as the density matrix of a system for which the states corresponding to the eigenvalues enclosed by $\mathcal{C}$ are all filled and all other states are empty. 
%\end{remark}

\subsection{Computation of upper left $2M \times 2M$ block of Green's function of SSH edge Hamiltonian with defects} \label{sec:GreenFunctionMethod_1}

For $z \in \mathbb{C} \setminus \sigma(\Hledge)$, the Green's function $G(z)$ of $\Hledged$ is defined as the bounded inverse of the operator $(z I - \Hledged)$ on $l^2(\mathbb{N};\mathbb{C}^2)$,
\begin{equation} \label{eq:def_G}
	( z I - \Hledged )G(z) = I.
\end{equation}
It is natural to label the entries of $G(z)$ similarly to those of the wavefunction $\psi$ so that the $m$th 2$\times$2 block of the part of $G(z)$ which acts on the $m'$th fundamental cell of $\psi$ is written
\begin{equation}
	\begin{bmatrix} G^{A,A}_{m,m'}(z) & G^{A,B}_{m,m'}(z) \\ G^{B,A}_{m,m'}(z) & G^{B,B}_{m,m'}(z) \end{bmatrix}.
\end{equation}
We now express \eqref{eq:def_G} column by column as follows. For notational clarity, we suppress $z$-dependence of entries of $G(z)$ and use the vector notation 
%so that we replace $G^{\sigma,\sigma'}_{m,m'}(z)$ by $G^\sigma_m$. We then further abbreviate by writing:
\begin{equation}
    G^{\sigma'}_{m,m'} = \begin{bmatrix} G_{m,m'}^{A,\sigma'} \\ G_{m,m'}^{B,\sigma'} \end{bmatrix}, \quad m \in \mathbb{N}.
\end{equation}
The defining equation of the Green's function \eqref{eq:def_G} then implies that the $\sigma',m'$th column of $G(z)$ satisfies 
\begin{equation} \label{eq:Greensfunc_sys}
\begin{split}
    &- A^*(m-1) G^{\sigma'}_{m-1,m'} + \left( z - V(m) \right) G^{\sigma'}_{m,m'} - A(m) G^{\sigma'}_{m+1,m'} = \delta^{\sigma,\sigma'}_{m,m'} \quad m \in \mathbb{N} \\
    &G^{\sigma'}_{0,m'} = 0,
\end{split}
\end{equation}
and the condition that $G^{\sigma'}_{m,m'} \rightarrow 0$ as $m \rightarrow \infty$ sufficiently fast that $\{ G^{\sigma'}_{m,m'} \}_{m \in \mathbb{N}} \in l^2(\mathbb{N},\mathbb{C}^2)$. Here the matrices $A(m)$ and $V(m)$ are as in \eqref{eq:ssh_edge_nontriv}, and $\delta^{\sigma,\sigma'}_{m,m'}$ denotes the vector with $1$ in its $m', \sigma'$th entry and $0$ in all its other entries. %\aw{notation changed here - column indices restored}

For $m' \leq M$ and $m \geq M + 2$, the system \eqref{eq:Greensfunc_sys} can be solved via a transfer matrix which depends on $z$
\begin{equation} \label{eq:G_trans}
	G^{\sigma'}_{m,m'} = T^\infty(z) G^{\sigma'}_{m-1,m'}, \quad T^\infty(z) := \begin{bmatrix} - \frac{t^\infty_1}{t^\infty_2} & \frac{z}{t^\infty_2} \\ - \frac{ z }{ t^\infty_2 } & \frac{ z^2 }{ t^\infty_1 t^\infty_2 } - \frac{t^\infty_2}{t^\infty_1} \end{bmatrix} \quad m \geq M + 2.
\end{equation}
Since the determinant of $T^\infty(z)$ is $1$ for all $z$, its spectrum is constrained as follows.
\begin{lemma} \label{lem:T_possibilities} Let $A$ be any $2 \times 2$ matrix with entries in $\mathbb{C}$. If the determinant of $A$ is $1$, the eigenvalues of $A$ satisfy precisely one of the following possibilities: 
\begin{enumerate}[label=(A\arabic*)]
\item $A$ has non-degenerate eigenvalues $\lambda$ and $\lambda^{-1}$ such that $|\lambda| < 1$ and $|\lambda^{-1}| > 1$
\item $A$ has non-degenerate eigenvalues $\lambda$ and $\lambda^{-1}$ such that $|\lambda| = |\lambda^{-1}| = 1$
\item $A$ has one degenerate eigenvalue equal to $1$ or $-1$.
\end{enumerate}
\end{lemma}
\begin{remark}
In the theory of periodic difference and differential equations (Hill's equations), the transfer matrix is known as the monodromy matrix and trichotomy (A1)-(A3) underlies Floquet's theorem \cite{MagnusWinkler,Eastham,CoddingtonLevinson}.
\end{remark}
Note that whenever $z$ is such that $T^\infty(z)$ has an eigenvalue of norm $1$ (i.e. either possibility (A2) or (A3) of Lemma \ref{lem:T_possibilities} is realized) we can produce a Weyl sequence and hence $z \in \sigma_{\text{ess}}(\Hledged)$. Hence the following Lemma.
\begin{lemma} \label{lem:eval_of_T_lem}
Whenever $z \notin \sigma(\Hbulk)$, where $\Hbulk$ is defined by \eqref{eq:ssh_bulk} with $t_1$ and $t_2$ replaced by $t_1^\infty$ and $t_2^\infty$, $T^\infty(z)$ realizes possibility (A1) of Lemma \ref{lem:T_possibilities}, i.e. has precisely one eigenvalue $\lambda$ such that $|\lambda| < 1$.
\end{lemma}
For the solution of \eqref{eq:Greensfunc_sys} to decay as $m \rightarrow \infty$ it must be that $G^{\sigma'}_{M+1,m'}$ is proportional to the associated eigenvector of this eigenvalue. Hence the following.
\begin{lemma}\label{lem:G_func_lem}
For $z \notin \sigma(\Hledge)$, $G^{\sigma'}_{M+1,m'}$, the $M+1$th entry of the $\sigma', m'$th column of the Green's function is proportional to the associated eigenvector of the eigenvalue of $T^\infty(z)$ with norm less than $1$ whose existence is guaranteed by Lemma \ref{lem:eval_of_T_lem}.
\end{lemma}

From Lemma \ref{lem:G_func_lem} we conclude that we can compute the first $M+1$ entries of the solution of \eqref{eq:Greensfunc_sys} (for $m' \leq M$) by truncating the system and imposing the boundary condition
\begin{equation} \label{eq:GF_BC}
	G^{\sigma'}_{M+1,m'}(z) = c \begin{bmatrix} \xi^A(z) \\ \xi^B(z) \end{bmatrix} \quad c \in \mathbb{C},
\end{equation}
where $[ \xi^A(z), \xi^B(z) ]^\trans$ denotes the eigenvector of $T^\infty(z)$ whose associated eigenvalue has norm less than one (whose existence is guaranteed by Lemma \ref{lem:eval_of_T_lem}) and $c$ denotes an arbitrary constant. In this way we can compute the first $M+1$ entries of the $(m',\sigma')$th column of $G(z)$ for each $m' \in \{1,...,M\}$ and $\sigma' \in \{A,B\}$ for any $z \notin \sigma(\Hledge)$. 

The choice to impose the boundary condition \eqref{eq:GF_BC} in the $M+1$th cell is to some extent arbitrary; the results of our method do not change if this boundary condition is imposed in any cell whose index is greater than $M$. However, choosing $M+1$ clearly minimizes the size of the matrix problem to solve and is hence optimal. For a numerical verification that our results do not depend on this choice see Figure \ref{fig:GF_verification}. 
\begin{figure}
\centerline{\includegraphics[scale=.75]{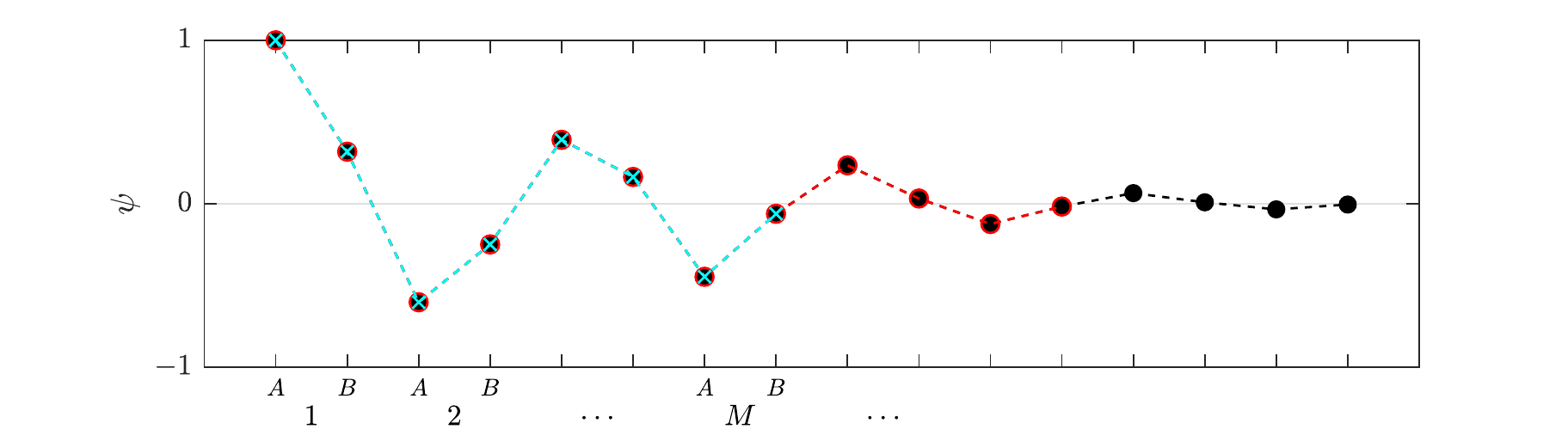}}
\caption{Plot of a bound state (eigenvalue $E = -.3837$) of the SSH edge Hamiltonian with defects $\Hledged$ \eqref{eq:ssh_edge_nontriv} (edge state) computed with the Green's function method for 3 different choices of truncation cell when $M = 4$. We plot the result of truncating in cells $M+1 = 5$ (cyan $\times$s), $7$ (red circles), and $9$ (black filled circles). For $m = 1$ to $4$, the three computations agree exactly. In this example $t_1^\infty = 1$ and $t_2^\infty = 2$ and we model disorder by drawing values of the hopping amplitudes $t_j(m), j \in \{1,2\}$, and onsite potentials $V^\sigma(m), \sigma \in \{A,B\}$ from normal distributions with mean equal to their values for $m \geq M$ and standard deviation $.5$.} 
\label{fig:GF_verification}
\end{figure}
\begin{remark}
Under Assumption \ref{as:bulk_periodicity}, the eigenequation $\Hledge \psi = E \psi$ can be solved directly using a transfer matrix similarly to the system \eqref{eq:Greensfunc_sys}. $E$ is then an eigenvalue of $\Hledge$ whenever the solution of this system is such that the $M$th entry of $\psi$ is proportional to an eigenvector of the transfer matrix whose associated eigenvalue has norm $< 1$. However this approach leads to a eigenvalue problem whose boundary condition (the analog of \eqref{eq:GF_BC}) depends on $E$ and is hence nonlinear. Such problems can be solved by iteration. We do not consider this idea further since it is inferior to the Green's function method.
\end{remark}

It remains to recover the eigenvalues and eigenfunctions of $\Hledged$ from the upper left $2 (M + 1) \times 2 M$ block of $G(z)$; we discuss this in Section \ref{sec:getStatesFromGamma}. In fact there is no loss in at this point discarding the $2 (M+1)$th row of this block and in what follows we will consider only the upper left $2 M \times 2 M$ block of $G(z)$.

\subsection{Recovery of spectral data from the truncated Green's function} \label{sec:getStatesFromGamma}

We now discuss how to recover the eigenvalues and eigenstates of $\Hledged$ from the truncated Green's function computed in the previous section. The upper left $2 M \times 2 M$ block of the spectral projection $\Gamma$ may be found by numerically evaluating the contour integral \eqref{eq:Gam} using the upper-left $2 M \times 2 M$ block of $G(z)$ computed in the previous section (see \cite{hale2008computing} for methods which dramatically speed up the evaluation of this integral). For notational simplicity, let us define
\begin{equation}
	\Gammatrunc := \text{ upper-left $2 M \times 2 M$ block of $\Gamma$. }
\end{equation}
The truncated spectral projection, $\Gammatrunc$, is not sufficient for recovering all eigenvalues and eigenfunctions of $\Hledged$. The missing information can be recovered by also using the upper-left $2 M \times 2 M$ block of the  operator $\Hledged \Gamma$, which may be computed via a contour integral, similarly to $\Gamma$.
\begin{equation} \label{eq:H_Gam}
	\Hledged \Gamma = \frac{1}{2 \pi i} \int_\mathcal{C} \Hledged G(z) \dd z = \frac{1}{2 \pi i} \int_\mathcal{C} z G(z) \dd z, \quad G(z) = (z - \Hledged)^{-1},
\end{equation}
where the second equality is clear by analyticity of the operator $(z - \Hledged ) G(z) = I$. The upper left $2M \times 2M$ block of $\Hledged \Gamma$ may now be computed similarly to that of $\Gamma$. We define
\begin{equation}
	\HGammatrunc := \text{ upper left $2 M \times 2 M$ block of $\Hledged \Gamma$}.
\end{equation}
Let $\psi_j$, where $j \in \{1,...,K\}$, denote the bound states of $\Hledge$ with associated eigenvalues $E_j$ within the contour $\gamma$ and $\psijtrunc$ the truncations of these bound states to the first $M$ cells. Then, using Dirac notation for simplicity, we have that  
\begin{equation}
\begin{split}
    &\Gamma = \sum_{j = 1}^K \ket{\psi_j}\bra{\psi_j} \implies \Gammatrunc = \sum_{j = 1}^K \ket{\psijtrunc}\bra{\psijtrunc}, \\
    &H \Gamma = H \sum_{j = 1}^K \ket{\psi_j}\bra{\psi_j} = \sum_{j = 1}^K E_j \ket{\psi_j}\bra{\psi_j} \\
    &\quad \quad \quad \quad \quad \quad \quad \implies \HGammatrunc = \sum_{j = 1}^K E_j \ket{\psijtrunc}\bra{\psijtrunc}.
\end{split}
\end{equation}
Our goal now is to recover each  $E_j$ and corresponding $\psijtrunc$ from the matrices $\Gammatrunc$ and $\HGammatrunc$. We consider in detail the case $K = 2$; the result for general $K$ is an easy generalization.

Let $v_1, v_2$ denote any basis (not necessarily orthogonal) of the range of $\Gammatrunc$. Since the matrices $\Gammatrunc$ and $\HGammatrunc$ map this space to itself, we can represent their action on vectors in the space $v = c_1 v_1 + c_2 v_2$ by matrices as follows
\begin{equation}
    \Gammatrunc \begin{bmatrix} c_1 \\ c_2 \end{bmatrix} = \begin{bmatrix} \alpha_{11} & \alpha_{12} \\ \alpha_{21} & \alpha_{22} \end{bmatrix} \begin{bmatrix} c_1 \\ c_2 \end{bmatrix}  \quad \HGammatrunc \begin{bmatrix} c_1 \\ c_2 \end{bmatrix} = \begin{bmatrix} \beta_{11} & \beta_{12} \\ \beta_{21} & \beta_{22} \end{bmatrix} \begin{bmatrix} c_1 \\ c_2 \end{bmatrix}.
\end{equation}
When we choose the basis $v_j = \psijtrunc$ then using the definitions of $\Gammatrunc$ and $\HGammatrunc$ we have that \begin{equation} \label{eq:psi_basis}
    \alpha_{i j} = \ip{\psi_i}{\psi_j}, \quad \beta_{i j} = E_i \ip{\psi_i}{\psi_j} \quad i,j \in \{1,2\}.
\end{equation}
Equation \eqref{eq:psi_basis} implies that the eigenvalues $E_j$ are exactly the eigenvalues of the matrix $M$ which relates the coefficients $\alpha_{ij}$ and $\beta_{ij}$, which satisfies
\begin{equation}
    M \begin{bmatrix} \alpha_{11} & \alpha_{12} \\ \alpha_{21} & \alpha_{22} \end{bmatrix}^\trans = \begin{bmatrix} \beta_{11} & \beta_{12} \\ \beta_{21} & \beta_{22} \end{bmatrix}^\trans,
\end{equation}
where $\phantom{}^\trans$ denotes matrix transpose. The existence of such a matrix is guaranteed by the fact that $\Gammatrunc$ is non-singular on its range, so that $M$ can be defined by
\begin{equation}
    M = \begin{bmatrix} \beta_{11} & \beta_{12} \\ \beta_{21} & \beta_{22} \end{bmatrix}^\trans \begin{bmatrix} \alpha_{11} & \alpha_{12} \\ \alpha_{21} & \alpha_{22} \end{bmatrix}^{-\trans},
\end{equation}
where $\phantom{}^{-\trans}$ denotes the inverse of the matrix transpose. With the eigenvalues $E_j$ in hand, the $\psijtrunc$ can then be found by computing each of the null spaces of the matrices $\HGammatrunc - E_j$. %\aw{fixed method here}

\section{Green's function method for arbitrary eventually periodic finite range Hamiltonians} \label{sec:trans_mat_in_general}

In this section we describe how the basic idea introduced in Section \ref{sec:GreenFunctionMethod} generalizes in a natural way to a class of one-dimensional discrete Hamiltonians which are finite range and eventually periodic. This generalization is necessary for us to consider the models of honeycomb structure edges we consider in Section \ref{sec:appToHoneycomb}. We start with the following definition.
\begin{definition} \label{def:ev_per}
We call any operator $H$ which is self-adjoint on $l^2(\mathbb{N};\mathbb{C}^N)$ a \ul{semi-infinite one-dimensional discrete Hamiltonian}. In addition: 
\begin{itemize}
\item We say $H$ is \ul{finite range} if there exists a non-negative integer $R$ such that the action of $H$ on elements $\psi$ of $l^2(\mathbb{N};\mathbb{C}^N)$ can be written in terms of the components $\psi_m \in \mathbb{C}^N$ of $\psi$ as
\begin{equation} \label{eq:general_Hed}
\begin{split}
    \left[ \Hedge \psi \right]_m &= A^*_R(m-R) \psi_{m-R} + ... + A^*_1(m-1) \psi_{m-1}  \\
    &\phantom{ = 20 } + V(m) \psi_m + A_1(m) \psi_{m+1} + ... + A_{R}(m) \psi_{m+R}
\end{split}
\end{equation}
for complex $N \times N$ matrices $A_j(m)$ where $j \in \{1,...,R\}$, and Hermitian $N \times N$ matrices $V(m)$ where $m \in \mathbb{N}$. Here $A^* := \overline{ A^\trans }$ denotes the conjugate transpose of the matrix $A$.
\item We say $H$ is \ul{eventually periodic} if there exists a non-negative integer $M$ such that
\begin{equation} \label{eq:even_periodic}
    V(m + 1) = V(m) \text{ and } A_j(m + 1) = A_j(m), \text{ $j \in \{1,...,R\}$ }
\end{equation}
for all $m \in \mathbb{N}$ such that $m \geq M$.
\end{itemize}
If $H$ is both finite range and eventually periodic we will usually abbreviate and call such a Hamiltonian a \ul{semi-infinite eventually periodic Hamiltonian}.
\end{definition}
Note that the SSH edge Hamiltonian with nearest neighbor hopping which we considered in Sections \ref{sec:SSH_pollution} and \ref{sec:GreenFunctionMethod} is a semi-infinite eventually periodic Hamiltonian with $R = 1$ and $N = 2$. 

\begin{remark}
By enlarging the fundamental cell sufficiently we can always take $R = 1$ in \eqref{eq:general_Hed}. However, we prefer not to do this so that we keep  $1$ as the minimal period of the models we consider. 
\end{remark}

It is important to note that our method can be generalized further and applied to discrete Hamiltonians defined on $\mathbb{Z}$ under appropriate conditions. %We consider an example of this kind in Section \ref{sec:domainWall}.

\begin{definition} \label{def:ev_per_2} We call any operator $H$ which is self-adjoint on $l^2(\mathbb{Z};\mathbb{C}^N)$ an \ul{infinite one-dimensional discrete Hamiltonian}. In addition: 
\begin{itemize}
\item We call $H$ \ul{finite range} if its action on elements $\psi$ of $l^2(\mathbb{Z};\mathbb{C}^N)$ can be written in terms of the components $\psi_m \in \mathbb{C}^N$ of $\psi$ as
\begin{equation} \label{eq:general_Hed_dw}
\begin{split}
    \left[ \Hedge \psi \right]_m &= A^*_R(m-R) \psi_{m-R} + ... + A^*_1(m-1) \psi_{m-1}  \\
    &\phantom{ = 20 } + V(m) \psi_m + A_1(m) \psi_{m+1} + ... + A_{R}(m) \psi_{m+R}
\end{split}
\end{equation}
for complex $N \times N$ matrices $A_j(m)$ where $j \in \{1,...,R\}$, and Hermitian $N \times N$ matrices $V(m)$ where $m \in \mathbb{Z}$. Here $A^* := \overline{ A^\trans }$ denotes the conjugate transpose of the matrix $A$.
\item We call $H$ \ul{eventually periodic} if there exist integers $M_+$ and $M_-$ such that $M_+ \geq M_-$ and 
\begin{equation} \label{eq:even_periodic_dw}
    V(m + 1) = V(m) \text{ and } A_j(m + 1) = A_j(m), \text{ $j \in \{1,...,R\}$ }
\end{equation}
for all $m \in \mathbb{Z}$ such that $m \geq M_+$, and
\begin{equation}
    V(m - 1) = V(m) \text{ and } A_j(m - 1) = A_j(m), \text{ $j \in \{1,...,R\}$ }
\end{equation}
for all $m \in \mathbb{Z}$ such that $m \leq M_-$.
\end{itemize}
If $H$ is both finite range and eventually periodic we will usually abbreviate and call such a Hamiltonian an \ul{infinite eventually periodic Hamiltonian}.
\end{definition}
In this case the method is almost identical except that when computing the Green's function the problem must be truncated in both directions with boundary conditions imposed at each end. We consider an example of this type in Section \ref{sec:domainWall}.

Since the method of recovering spectral data from the Green's function in general is identical to that given in Section \ref{sec:getStatesFromGamma} and the method is similar for infinite eventually periodic Hamiltonians we focus in this section on the problem of computing the upper left block of the Green's function of a general semi-infinite eventually periodic Hamiltonian.

%The case which we considered in Section \ref{sec:GreenFunctionMethod} was relatively simple because the difference equation satisfied by the Green's function could be solved using a $2 \times 2$ transfer matrix, and hence the correct boundary condition to impose at the truncation was clear.

%Recall that in Section \ref{sec:GreenFunctionMethod} we showed that the discrete spectrum and associated bound states of a semi-infinite Hamiltonian $H$ could be accurately recovered as long as the Green's function of $H$, i.e. the solution of
%\begin{equation} \label{eq:GF}
    %(z - H) G(z) = I
%\end{equation}
%could be accurately computed for $z$ along any contour enclosing the desired subset of spectrum. To solve \eqref{eq:GF}, we restricted attention to computing the first $M$ entries of $G(z)$, knowing that the remainder of $G(z)$ could be found by application of a transfer matrix, using the transfer matrix to define an appropriate boundary condition on the computation in the $M$th cell. 

\subsection{Constructing boundary conditions for  Green's functions computation for semi-infinite eventually periodic Hamiltonian} \label{sec:BCs_in_general}

For values of $z \in \mathbb{C}$ in the resolvent set, the Green's function $G(z)$ of an arbitrary semi-infinite eventually periodic Hamiltonian $H$ is defined as the bounded inverse of $z I - H$ on $l^2(\mathbb{N};\mathbb{C}^N)$,
\begin{equation} \label{eq:general_GF}
    (z I - H) G(z) = I.
\end{equation}
Just as in Section \ref{sec:GreenFunctionMethod}, we represent the $m$th $N\times N$ block of $G(z)$ which acts on the $m'$th fundamental cell as
\begin{equation}
    \begin{bmatrix} G^{A,A}_{m,m'}(z) & G^{A,B}_{m,m'}(z) & G^{A,C}_{m,m'}(z) & ... \\ G^{B,A}_{m,m'}(z) & G^{B,B}_{m,m'}(z) & G^{B,C}_{m,m'}(z) & ... \\ G^{C,A}_{m,m'}(z) & G^{C,B}_{m,m'}(z) & G^{C,C}_{m,m'}(z) & ... \\ ... & ... & ... & ... \end{bmatrix}.
\end{equation}
Equation \eqref{eq:general_GF} can then be written column-by-column as follows. We again suppress $z$-dependence of entries of $G(z)$ and then assemble vectors 
\begin{equation}
    G^{\sigma'}_{m,m'} = \begin{bmatrix} G^{A,\sigma'}_{m,m'} , G^{B,\sigma'}_{m,m'} , ... \end{bmatrix}^\trans
\end{equation}
for each $m \in \mathbb{Z}$. We have then that the entries of the $m', \sigma'$th column of the Green's function satisfy the system 
\begin{equation} \label{eq:G_generall}
\begin{split}
    &- A^*_R(m-R) G^{\sigma'}_{m-R,m'} - ... - A^*_1(m-1) G^{\sigma'}_{m-1,m'}   \\
    &\phantom{ + bla bla bla bla }+ \left( z - V(m) \right) G^{\sigma'}_{m,m'} - A_1(m) G^{\sigma'}_{m+1,m'} - ... - A_R(m) G^{\sigma'}_{m+R,m'} = \delta_{m,m'}^{\sigma,\sigma'}    \\
    &G^{\sigma'}_{m,m'} = 0 \quad m \leq 0,
\end{split}
\end{equation}
and the condition that $G^{\sigma'}_{m,m'} \rightarrow 0$ as $m \rightarrow \infty$ sufficiently fast that $\{ G^{\sigma'}_{m,m'} \}_{m \in \mathbb{N}} \in l^2(\mathbb{N};\mathbb{C}^2)$. Here the matrices $A^*_J(m)$ are those which appear in \eqref{eq:general_Hed}, and $\delta^{\sigma,\sigma'}_{m,m'}$ denotes the vector whose $\sigma', m'$th entry is one and all others are zero. Since we require only the first $M$ columns of the Green's function we need only solve \eqref{eq:G_generall} for $m' \leq M$. 

For clarity we at this point restrict to the case $R = 1$ (see Remark \ref{rem:other_R} for the case where $R > 1$), and drop the subscript on $A_1 \mapsto A$ so that \eqref{eq:G_generall} reduces to %\jl{perhaps in this case we can drop the subscript $1$ in $A$} \aw{fixed}
\begin{equation} \label{eq:G_generall_2}
\begin{split}
    &- A^*(m-1) G^{\sigma'}_{m-1,m'} + \left( z - V(m) \right) G^{\sigma'}_{m,m'} - A(m) G^{\sigma'}_{m+1,m'} = \delta_{m,m'}^{\sigma,\sigma'} \quad m \in \mathbb{N},   \\
    &G^{\sigma'}_{0,m'} = 0,
\end{split}
\end{equation}
together with the condition that $G^{\sigma'}_{m,m'} \rightarrow 0$ as $m \rightarrow \infty$. Rather than study \eqref{eq:G_generall_2} directly, we will study the model difference equation for a general right hand side 
\begin{equation} \label{eq:G_generall_3}
\begin{split}
    &- A^*(m-1) g_{m-1} + \left( z - V(m) \right) g_m - A(m) g_{m+1} = f_m \quad m \in \mathbb{N}, \\
    &g_0 = 0,
\end{split}
\end{equation}
where our only assumption on the $f_m$ is that $f_m = 0$ for $m > M$ and we impose $g_m \rightarrow 0$ as $m \rightarrow \infty$. Our goal is to find appropriate boundary conditions which will allow us to truncate \eqref{eq:G_generall_3} to a finite problem whose solution can be computed.

Under the eventual periodicity assumption \eqref{eq:even_periodic}, for $m \geq M + 2$ the system \eqref{eq:G_generall_3} can be written redundantly as
\begin{equation} \label{eq:genrlized}
    S^*(\overline{z}) \begin{bmatrix} g_{m-2} \\ g_{m-1} \end{bmatrix} + S(z) \begin{bmatrix} g_{m} \\ g_{m+1} \end{bmatrix} = 0 \quad m \geq M + 2,
\end{equation}
%\jl{check the range of $m$, $m\geq M+3$ or $m\geq M+2$?} \aw{fixed}
where
\begin{equation} \label{eq:S_of_z}
    S(z) := \begin{bmatrix} - A^\infty & 0 \\ z - V^\infty & - A^\infty \end{bmatrix}.
\end{equation}
Here $V^\infty := V(M)$ and $A^\infty := A(M)$. We will prove the following theorem, which generalizes the basic ideas developed in Section \ref{sec:GreenFunctionMethod}.
\begin{theorem} \label{th:BCs_th}
Let $z \notin \sigma(H)$ and let Assumption \ref{as:M_invertible} (see below) on the matrix $S(z)$ \eqref{eq:S_of_z} hold. Then there exists an $N$-dimensional subspace of $\mathbb{C}^{2 N}$, denoted in what follows by $V_{\text{decaying}}$, such that entries $g_1,g_2,...,g_{M+2}$ of the solution of the system \eqref{eq:G_generall_3} subject to $\left\{ g_m \right\}_{m \in \mathbb{N}} \in l^2(\mathbb{N};\mathbb{C}^N)$ can be found by truncating the system after $M + 2$ cells and imposing the boundary condition that $[g_{M+1},g_{M+2}]^\trans \in V_{\text{decaying}}$.
\end{theorem}
The proof makes use of the following Lemma
\begin{lemma} \label{lem:symp}
Let $\mathcal{S}$ denote an invertible $2 \mathcal{N} \times 2 \mathcal{N}$ complex matrix and let $\mathcal{T} := \mathcal{S}^{-1} \mathcal{S}^*$. Then the characteristic polynomial of $\mathcal{T}$ satisfies the symmetry 
\begin{equation}
    \det[ \mathcal{T} - \lambda I ] = \lambda^{2 \mathcal{N}} \det\left[ \mathcal{T} - \frac{1}{\overline{\lambda}} I \right] \quad \lambda \in \mathbb{C} \setminus {0},
\end{equation}
and hence whenever $\lambda \in \mathbb{C} \setminus {0}$ is an eigenvalue of $\mathcal{T}$ with multiplicity $\mu$ so is $\frac{1}{\overline{\lambda}}$.
\end{lemma}
\begin{proof}[Proof of Lemma \ref{lem:symp}]
The proof is by direct computation using standard identities satisfied by the determinant:
\begin{equation}
\begin{split}
    \det[ \mathcal{S}^{-1} \mathcal{S}^* - \lambda I ] &= \lambda^{2 \mathcal{N}} \det[ \mathcal{S}^{-1} ] \det\left[ \frac{1}{ \lambda } I - \mathcal{S} ( \mathcal{S}^* )^{-1} \right] \det[ \mathcal{S}^* ]  \\
    &= \lambda^{2 \mathcal{N}} \det\left[ \frac{1}{\overline{\lambda}} I - \mathcal{S}^{-1} \mathcal{S}^* \right].
\end{split}
\end{equation}
\end{proof}
We now prove Theorem \ref{th:BCs_th}. 
\begin{proof}[Proof of Theorem \ref{th:BCs_th}]
In the simplest case, $S(z)$ is non-singular and we can obtain a transfer matrix relating $g_{m-2}, g_{m-1}$ with $g_{m}, g_{m+1}$
\begin{equation} \label{eq:traans}
    \begin{bmatrix} g_{m} \\ g_{m+1} \end{bmatrix} = T^\infty(z) \begin{bmatrix} g_{m-2} \\ g_{m-1} \end{bmatrix} \quad T^\infty(z) := S^{-1}(z) S^*(\overline{z}).
\end{equation}
By assumption, $z \notin \sigma(H)$. If $T^\infty(z)$ has an eigenvalue of magnitude $1$ we can generate a Weyl sequence which is a contradiction. On the other hand, when $z$ is real we may apply Lemma \ref{lem:symp} to the matrix $T^\infty(z)$ to see that $T^\infty(z)$ must have the same number of (possibly generalized) eigenvectors corresponding to eigenvalues with norm less than $1$ as corresponding to eigenvalues with norm greater than $1$, and hence must have precisely $N$ (possibly generalized) eigenvectors corresponding to eigenvalues with norm less than $1$. The statement of the theorem now follows for such $z$ by setting $V_{\text{decaying}}$ equal to the span of these vectors.

When $z \notin \sigma(H)$ is complex, the conclusion still holds by the following argument. Consider a smooth path between any real $z_0 \notin \sigma(H)$ and $z$ which does not intersect with $\sigma(H)$. Then the number of eigenvalues of $T^\infty(z)$ with norm less than $1$ cannot change along this path since if it did at the point of change an eigenvalue would have to lie on the unit circle in the complex plane. But we have a contradiction since this point must then be in the spectrum of $H$ since we can produce a Weyl sequence.

Often, however, the matrix $S(z)$ is singular. For the SSH model with nearest-neighbor hopping for example, $S(z)$ has the form
\begin{equation} \label{eq:SSH_S}
    \begin{bmatrix} 0 & 0 & 0 & 0 \\ - t_2^\infty & 0 & 0 & 0 & \\ z & - t^\infty_1 & 0 & 0 \\ - t^\infty_1 & z & - t_2^\infty & 0 \end{bmatrix},
\end{equation}
which has an obvious null space spanned by $[0,0,0,1]^\trans$. In such a case the system can still be solved using a transfer matrix, but the derivation is more involved. We require the following assumption which is clearly satisfied by \eqref{eq:SSH_S}.
\begin{assumption} \label{as:M_invertible}
We assume that if $S(z)$ is not invertible, it has the form 
\begin{equation}
    S(z) = \begin{bmatrix} 0 & 0 & 0 \\ M_{11} & M_{12} & 0 \\ M_{21} & M_{22} & 0 \end{bmatrix}
\end{equation}
where the complex matrices $M_{11}, M_{12}, M_{21}, M_{22}$ have sizes $2 m \times n$, $2 m \times 2 m$, $n \times n$, $n \times 2 m$, where $m := N - n > 0$. The matrices $M_{21}$ and $M_{12} - M_{11} (M_{21})^{-1} M_{22}$, which have sizes $n \times n$ and $2 m \times 2 m$ respectively, are invertible.
\end{assumption}
Under Assumption \ref{as:M_invertible}, the matrix
\begin{equation}
    M = \begin{bmatrix} M_{11} & M_{12} \\ M_{21} & M_{22} \end{bmatrix}
\end{equation}
is invertible by Schur complement and hence $n$ is precisely the dimension of $\text{Null } S(z)$. Assuming again that $z$ is real, the conjugate transpose $S^*(z)$ has the form
\begin{equation}
    S^*(z) = \begin{bmatrix} 0 & M_{11}^* & M_{21}^* \\ 0 & M_{12}^* & M_{22}^* \\ 0 & 0 & 0 \end{bmatrix},
\end{equation}
and Null $S^*(z) \perp $ Null $S(z)$. We may therefore introduce orthogonal projections onto these spaces and their orthogonal complement
\begin{equation}
    P_{\text{ Null }S^*(z)}, \quad P_{\text{ Null }S(z)}, \quad P^\perp(z) := 1 - P_{\text{ Null }S^*(z)} - P_{\text{ Null }S(z)}.
\end{equation}
For simplicity we now suppress dependence on $z$ and introduce the notation
\begin{equation}
    A_n := P_{\text{ Null }S^*} \begin{bmatrix} g_{n} \\ g_{n+1} \end{bmatrix}, B_n := P^\perp \begin{bmatrix} g_{n} \\ g_{n+1} \end{bmatrix}, C_n := P_{\text{ Null }S} \begin{bmatrix} g_{n} \\ g_{n+1} \end{bmatrix}.
\end{equation}
The system \eqref{eq:genrlized} can then be written row by row as
\begin{equation} \label{eq:r_b_r}
\begin{split}
    &M_{11}^* B_{m-2} + M_{21}^* C_{m-2} = 0    \\
    &M_{12}^* B_{m-2} + M_{22}^* C_{m-2} + M_{11} A_m + M_{12} B_m = 0   \quad m \geq M + 2 \\
    &M_{21} A_m + M_{22} B_m = 0.
\end{split}
\end{equation}
Using invertibility of $M_{21}$ (Assumption \ref{as:M_invertible}) the first and third equations of the system can be written 
\begin{equation} \label{eq:solvability}
    C_{m-2} = - (M_{21}^*)^{-1} M_{11}^* B_{m-2} 
\end{equation}
\begin{equation} \label{eq:A_m_eq}
    A_m = - (M_{21})^{-1} M_{22} B_m.
\end{equation}
Substituting these relations into the second equation yields
\begin{equation}
    ( M_{12}^* - M_{22}^* (M_{21}^*)^{-1} M_{11}^* ) B_{m-2} + ( M_{12} - M_{11} (M_{21})^{-1} M_{22} ) B_m = 0.
\end{equation}
By Assumption \ref{as:M_invertible} we now see that
\begin{equation} \label{eq:B_m_eq}
    B_m = T B_{m-2},
\end{equation}
where
\begin{equation} \label{eq:new_T}
    T := - ( M_{12} - M_{11} (M_{21})^{-1} M_{22} )^{-1} ( M_{12}^* - M_{22}^* (M_{21}^*)^{-1} M_{11}^* ).
\end{equation}
Assuming that $B_{m-2}$ and $C_{m-2}$ satisfy the solvability condition \eqref{eq:solvability}, we see that the system \eqref{eq:r_b_r} uniquely specifies $B_m$ through \eqref{eq:B_m_eq}, which in turn specifies $A_m$ through \eqref{eq:A_m_eq}. In fact, $C_m$ is also uniquely specified by $B_m$ through the condition that the system \eqref{eq:r_b_r} be solvable when $m$ is replaced by $m+2$, i.e.
\begin{equation}
    C_m = - (M_{21}^*)^{-1} M_{11}^* B_{m}.
\end{equation}
It now follows that if $B_{m-2}$ lies in the associated eigenspace of the eigenvalues of \eqref{eq:new_T} with norm less than $1$, then the system \eqref{eq:r_b_r} has a solution which decays as $m \rightarrow \infty$. In terms of the original notation, we have that as long as (1) $\begin{bmatrix} g_{m-2}, g_{m-1} \end{bmatrix}^\trans$ satisfies the condition that
\begin{equation}
    P_{ \text{ Null }S} \begin{bmatrix} g_{m-2} \\ g_{m-1} \end{bmatrix} = - (M^*_{21})^{-1} M_{11}^* P^\perp \begin{bmatrix} g_{m-2} \\ g_{m-1} \end{bmatrix},
\end{equation}
and (2) $P^\perp \begin{bmatrix} g_{m-2}, g_{m-1} \end{bmatrix}^\trans$ lies in the image of the projection onto the space of associated eigenvectors of eigenvalues of $T$ \eqref{eq:new_T} with norm less than one, then the system \eqref{eq:genrlized} has a unique solution which lies in $l^2(\mathbb{N};\mathbb{C}^N)$. We therefore define $V_{\text{decaying}}$ as the space of vectors satisfying both of these conditions. By Lemma \ref{lem:symp} this space has dimension precisely $N$. The generalization to complex $z \notin \sigma(H)$ follows from the same argument given for the case where $S(z)$ is non-singular. 
\end{proof}

\begin{remark} \label{rem:other_R}
When $R > 1$, the discussion is the same except that the matrix $S$ must be taken larger. For example when $R = 2$, we take
\begin{equation}
    S = \begin{bmatrix} - A_2 & 0 & 0 & 0 \\ - A_1 & - A_2 & 0 & 0 \\ z - V & - A_1 & - A_2 & 0 \\ - A_1^* & z - V & - A_1 & - A_2 \end{bmatrix},
\end{equation}
which now acts on 4 cells of the column $\{ g_m \}_{m \in \mathbb{N}}$. %\jl{not sure if we need to discuss this as we said before that we can just assume $R = 1$} \aw{if we agree to not simply take $R = 1$ then it might help understanding}
\end{remark}

%It is illuminating to compute $T$ defined by \eqref{eq:new_T} for the SSH model. In this case we have
%\begin{equation}
    %M_{11} = \begin{bmatrix} - t_2^\infty \\ z \end{bmatrix}, M_{12} = \begin{bmatrix} 0 & 0 \\ - t_1^\infty & 0 \end{bmatrix}, M_{21} = - t_1^\infty, M_{22} = \begin{bmatrix} z & - t_2^\infty \end{bmatrix}
%\end{equation}
%and hence
%\begin{equation}
%\begin{split}
    %T &= - \begin{bmatrix} - \frac{ t_2^\infty z }{ t_1^\infty } & - \frac{ (t_2^\infty)^2 }{ t_1^\infty } \\ \frac{ z^2 }{ t_1^\infty } - t_1^\infty & \frac{z^2}{t_1^\infty} \end{bmatrix}^{-1} \begin{bmatrix} - \frac{ t_2^\infty z }{ t_1^\infty } &  \frac{ z^2 }{ t_1^\infty } - t_1^\infty  \\ - \frac{ (t_2^\infty)^2 }{ t_1^\infty } & \frac{z^2}{t_1^\infty} \end{bmatrix}   \\
    %&= \frac{1}{ - \frac{(t_2^\infty)^2}{t_1^\infty} +  }
%\end{split}
%\end{equation}

Assumption \ref{as:M_invertible} holds for all the models we consider in this paper; so, Theorem \ref{th:BCs_th} gives a relatively simple and direct proof of our method for these cases.  However, we note that Assumption \ref{as:M_invertible} is not necessary.  In Appendices \ref{sec:method_in_general} and \ref{sec:method_in_general_proofs}, we generalize our boundary condition method so that only two assumptions are required \dash the Hamiltonian must be (1) of finite range and (2) eventually periodic.  The more general algorithm is slightly different than the one given above due to some special situations that can occur in the more general setting.  See the appendices for the detailed general algorithm, as well as the full proof that the more general algorithm works in this general setting.

\section{Application to edge states of honeycomb structures} \label{sec:appToHoneycomb}

In this section, we apply our method to Hamiltonians which model electronic states at edges of two-dimensional honeycomb structures with defects. The structure of this section is as follows.

In Section \ref{sec:honeyc_bulk} we define the tight-binding bulk Hamiltonian of a general honeycomb structure with real, nearest-neighbor hopping. In Section \ref{sec:honey_zz_edge} we consider Hamiltonians describing states at a zig-zag edge of a graphene-like structure when the edge has defects. We choose to focus on the zig-zag edge because this edge (as opposed to, for example, an armchair edge) is known to support edge states in the absence of disorder. In Section \ref{sec:domainWall} we consider an edge Hamiltonian describing states localized along the interface of two distinct such structures. We assume in this case again that the interface occurs along a zig-zag edge. %\jl{perhaps we shall explain why we focus on zig-zag edges and how about other edges} \aw{fixed}
In both cases we employ a supercell-type approximation to reduce a two-dimensional problem to a one-dimensional problem. In Section \ref{sec:continuum} we show that our method can also be applied to a continuum Schr\"odinger equation model of the boundary of a honeycomb structure under finite difference discretization.

\subsection{Tight-binding honeycomb structure with nearest-neighbor hopping bulk Hamiltonian} \label{sec:honeyc_bulk}
We consider a single electron without spin hopping on a two-dimensional honeycomb lattice. Each fundamental cell, labeled by integers $m$ and $n$, of such a lattice hosts two atoms, which we label $A$ and $B$. We introduce the notation
\begin{equation}
	\psi_{m,n} = \begin{bmatrix} \psi^A_{m,n} \\ \psi^B_{m,n} \end{bmatrix},
\end{equation}
to denote the restriction of the electron wavefunction $\psi \in l^2(\mathbb{Z}^2;\mathbb{C}^2)$ to the $(m, n)$th fundamental cell of the lattice. 
The bulk Hamiltonian of a general honeycomb structure with nearest-neighbor hopping is:
\begin{equation} \label{eq:struc}
\begin{split}
	\left[ \Hbulk \psi \right]_{m,n} = \phantom{=} &A^* \psi_{m-1,n} + B^* \psi_{m,n-1}     \\
        &\phantom{blablabla}+ V \psi_{m,n} + A \psi_{m+1,n} + B \psi_{m,n+1} \quad m,n \in \mathbb{Z} \times \mathbb{Z}   \\
\end{split}
\end{equation}
\begin{equation}
        A := \begin{bmatrix} 0 & 0 \\ -t_1 & 0 \end{bmatrix}, \quad B := \begin{bmatrix} 0 & 0 \\ -t_2 & 0 \end{bmatrix}, \quad V := \begin{bmatrix} V^A & -t_0 \\ -t_0 & V^B \end{bmatrix},
        %- \begin{bmatrix} t_0 \psi^B_{m,n} + t_1 \psi^B_{m-1,n} + t_2 \psi^B_{m,n-1} \\ t_0 \psi^A_{m,n} + t_1 \psi^A_{m+1,n} + t_2 \psi^A_{m,n+1} \end{bmatrix} + \begin{bmatrix} V^A \psi^A_{m,n} \\ V^B \psi^B_{m,n} \end{bmatrix} \quad m,n \in \mathbb{Z} \times \mathbb{Z}.
\end{equation}
where the onsite potentials $V^A, V^B$ are assumed real, and the hopping amplitudes $t_0, t_1, t_2$ are assumed real and non-zero. Bloch's theorem in this context implies that bounded eigenfunctions of $\Hbulk$ satisfy: 
\begin{equation}
    \Phi_{m+1,n}(k_1,k_2) = e^{i k_1} \Phi_{m,n}, \quad \Phi_{m,n+1}(k_1,k_2) = e^{i k_2} \Phi_{m,n}
\end{equation}
for $k_j \in [-\pi,\pi], j \in \{1,2\}$. It is then a standard calculation to find the Bloch bands $E^\pm(k_1,k_2$) of $\Hbulk$. Depending on the values of the $t_j$ and $V^\sigma$ the Hamiltonian \eqref{eq:struc} may or may not have a spectral gap.

When $t_0 = t_1 = t_2$ and $V^\sigma = 0$, $\sigma \in \{A,B\}$ the Hamiltonian \eqref{eq:struc} models graphene whose Bloch bands are non-degenerate other than at $K := \left(\frac{2 \pi}{3},-\frac{2 \pi}{3}\right)$ and $- K$, known as the ``Dirac points'', where the bands touch at $E = 0$ \cite{2009Castro-NetoGuineaPeresNovoselovGeim}. 

\subsection{Tight-binding honeycomb structure edge Hamiltonians with defects} \label{sec:honey_zz_edge}
We now model the edge of a honeycomb structure possibly subject to defects as follows. We allow for inter-atom hopping amplitudes $t_j(m,n), j \in \{0,1,2\}$ and for onsite potentials $V^\sigma(m,n), \sigma \in \{A,B\}$ which depend on the cell indices $m,n$, making the following assumptions.
\begin{assumption}[Periodicity of bulk medium] \label{as:eventually_nodefects}
There exists a non-negative integer $M$ such that for all $m \geq M$, $V^{\sigma}(m,n) = V^\sigma$ for some real $V^\sigma$ and $t_j(m,n) = t_j$ for some real $t_j$, for all $n \in \mathbb{Z}$, $\sigma \in \{A, B\}$, and $j \in \{0,1,2\}$.
\end{assumption}
\begin{assumption}[Periodicity of disorder along edge] \label{as:edge_periodicity}
There exists a positive integer $N$ such that $V^\sigma(m,n+N) = V^{\sigma}(m,n)$ and $t_j(m,n+N) = t_j(m,n)$ for all $m \in \mathbb{N}$, $n \in \mathbb{Z}$,  $\sigma \in \{A, B\}$, and $j \in \{0,1,2\}$.
\end{assumption}
\begin{remark}
Our method may be used to study the effect of more general disorder not satisfying Assumption \ref{as:edge_periodicity} for any $N$. This leads to an approximate numerical method known as the supercell method (see for example \cite{AllenTildesley}). Under Assumption \ref{as:edge_periodicity} our method is exact. 
\end{remark}
In summary, we consider the Hamiltonian
\begin{equation*}
\begin{split}
        &\left[ \Hzzedged \psi \right]_{m, n} = \phantom{=} A^*(m-1,n) \psi_{m-1,n} + B^*(m,n-1) \psi_{m,n-1} \\
        & \phantom{blablablablablabla} + V(m,n) \psi_{m,n} + A(m,n) \psi_{m+1,n} + B(m,n) \psi_{m,n+1} \quad m, n \in \mathbb{N} \times \mathbb{Z}  \\
	&\psi_{0,n} = 0 \text{ for } n \in \mathbb{Z}.    \\
\end{split}
\end{equation*}
%	- \begin{bmatrix} t_0(m,n) \psi^B_{m,n} + t_1(m,n) \psi^B_{m-1,n} + t_2(m,n) \psi^B_{m,n-1} \\ t_0(m,n) \psi^A_{m,n} + t_1(m+1,n) \psi^A_{m+1,n} + t_2(m,n+1) \psi^A_{m,n+1} \end{bmatrix} + \begin{bmatrix} V^A(m,n) \psi^A_{m,n} \\ V^B(m,n) \psi^B_{m,n} \end{bmatrix} 	\\
%	&\qquad \qquad \qquad \qquad \; \qquad \qquad \qquad \qquad \qquad \qquad \qquad \qquad \qquad \qquad \qquad \qquad \qquad  m,n \in \mathbb{N}\times\mathbb{Z}	\\
where the matrices $A(m,n), B(m,n), V(m,n)$ are defined by 
\begin{equation} \label{eq:hop_matrices}
\begin{aligned}
        & A(m,n) = \begin{bmatrix} 0 & 0 \\ t_1(m,n) & 0 \end{bmatrix} \quad B(m,n) = \begin{bmatrix} 0 & 0 \\ 0 & t_2(m,n) \end{bmatrix} \\
        & V(m,n) = \begin{bmatrix} V^A(m,n) & t_0(m,n) \\ t_0(m,n) & V^B(m,n) \end{bmatrix} 
        \end{aligned}
\end{equation}
and the $V^\sigma(m,n)$, $\sigma \in \{A,B\}$, and $t_j(m,n)$, $j \in \{0,1,2\}$ are as in Assumptions \ref{as:eventually_nodefects} and \ref{as:edge_periodicity}.
\begin{remark}
A considerable literature related with edge states of graphene exists, mostly focusing on graphene ``ribbons'', see e.g., \cite{1996NakadaFujitaDresselhausDresselhaus,2009Castro-NetoGuineaPeresNovoselovGeim,2011DelplaceUllmoMontambaux}.
\end{remark}
Under Assumption \ref{as:edge_periodicity}, there is no loss of generality in imposing the quasi-periodic boundary condition (Bloch transforming) with respect to $n$ %\jl{this sentence seems a bit confusing; I thought you meant to say that we can do a Bloch transform in the $n$ direction, perhaps we shall state this more explicitly} \aw{added clarification}
\begin{equation} \label{eq:edge_quasi_per}
	\psi_{m,n+N} = e^{i k_\parallel N} \psi_{m,n}, \quad m,n \in \mathbb{N} \times \mathbb{Z},
\end{equation}
where $k_\parallel \in \left[-\frac{\pi}{N},\frac{\pi}{N}\right]$, which leads to the $k_\parallel$-dependent eigenvalue problem defined on a single supercell
\begin{equation} \label{eq:supercell_edge_H}
\begin{split}
        &\left[ \Hzzedged \psi(k_\parallel) \right]_{m, n} = \phantom{=} A^*(m-1,n) \psi_{m-1,n} + B^*(m,n-1) \psi_{m,n-1} + V(m,n) \psi_{m,n}  \\
        & \phantom{blablablablablablablabla} + A(m,n) \psi_{m+1,n} + B(m,n) \psi_{m,n+1} \quad m, n \in \mathbb{N} \times \left\{1,...,N\right\}  \\
%	&\psi_{0,n} = 0 \text{ for } n \in \left\{1,...,N\right\}.    \\
%	&\left[ \Hzzedged(k_\parallel) \psi \right]_{m, n} = - \begin{bmatrix} t_0(m,n) \psi^B_{m,n} + t_1(m-1,n) \psi^B_{m-1,n} + t_2(m,n-1) \psi^B_{m,n-1} \\ t_0(m,n) \psi^A_{m,n} + t_1(m,n) \psi^A_{m+1,n} + t_2(m,n) \psi^A_{m,n+1} \end{bmatrix} + \begin{bmatrix} V^A(m,n) \psi^A_{m,n} \\ V^B(m,n) \psi^B_{m,n} \end{bmatrix} \\
%	&\; \quad \qquad \qquad \qquad \qquad \qquad \qquad \qquad \qquad \qquad \qquad \qquad \qquad \qquad \qquad \qquad  m,n \in \mathbb{N}\times\{1,...,N\}	\\
%	&V^{\sigma}(m,n) = 0 \text{ and } t_j(m,n) = t \text{ for } m \in \mathbb{N} \text{ and } m > M, n \in \{1,...,N\}, j \in \{0,1,2\}, \sigma \in \{A,B\} \\
%	&t_2(m,0) = t_2(m,N) \text{ for } m \in \mathbb{N}	\\
	&\psi_{0,n} = 0 \text{ for } n \in \{1,...,N\} \quad \psi_{m,0} = e^{- i k_\parallel N} \psi_{m,N} \quad \psi_{m,N+1} = e^{i k_\parallel N} \psi_{m,1}, \text{ for } m \in \mathbb{N}.
\end{split}
\end{equation}
The system \eqref{eq:supercell_edge_H} is now an eventually periodic semi-infinite Hamiltonian in the sense of Definition \ref{def:ev_per} defined on $l^2(\mathbb{N};\mathbb{C}^{2 N})$ by defining
\begin{equation} \label{eq:honeycomb_notation}
    \psi_m^\sigma := \left[ \psi^\sigma_{m,1} , \psi^\sigma_{m,2} , ... \right]^\trans \quad \sigma \in \{A,B\},
\end{equation}
and is hence amenable to our method. An illustration of this notation is given in Figure \ref{fig:wavefunction_notation}, while results of computations using our method are shown in Figure \ref{fig:ZZ_remove_atom}.

\begin{figure}
\centerline{\includegraphics[scale=.45]{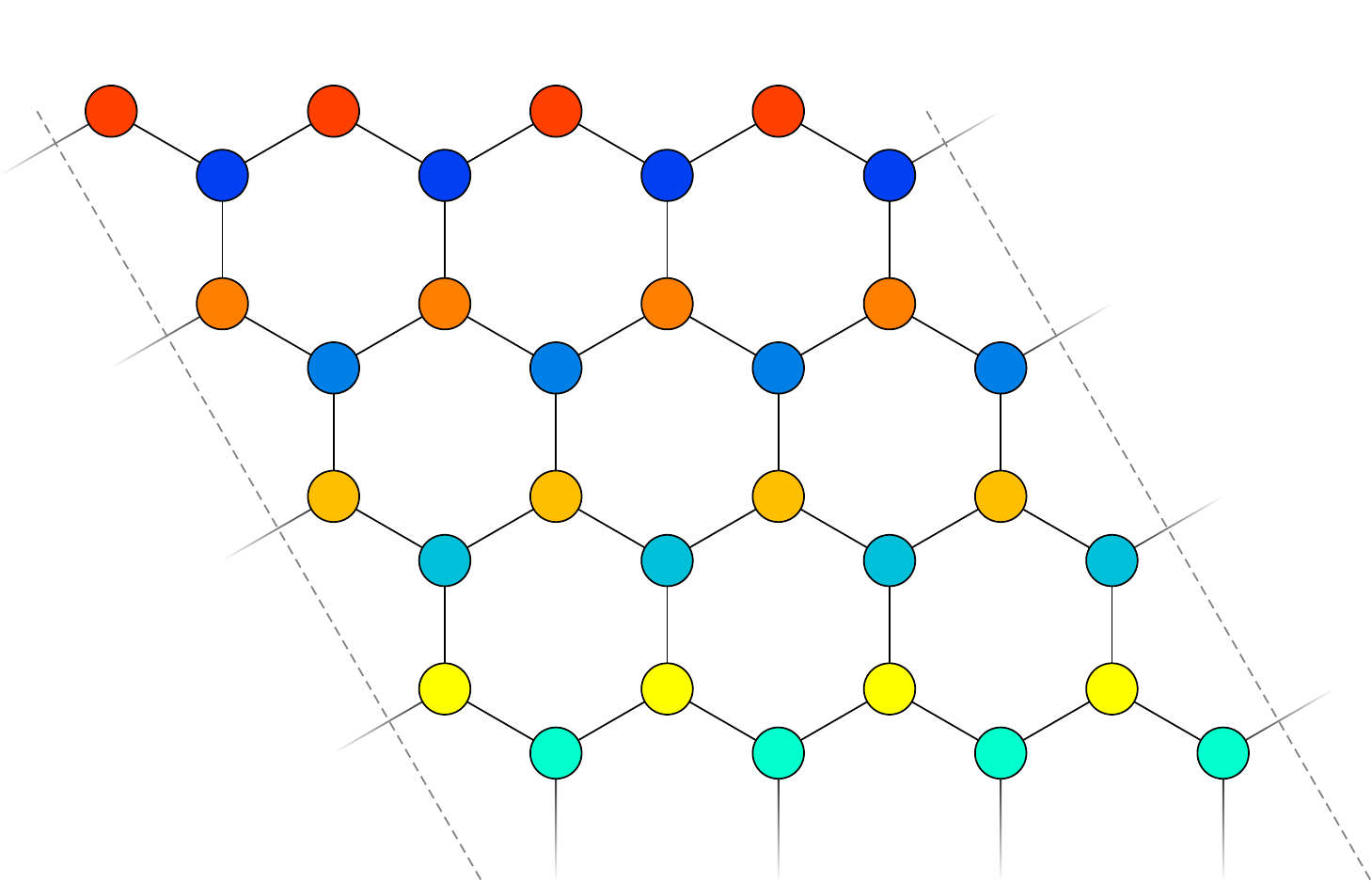}}
\caption{This figure illustrates the notation defined in \eqref{eq:honeycomb_notation}.  The coloring defines the different $\vec{\psi}^\sigma_m$ as follows: \textcolor[rgb]{1,0,0}{$\vec{\psi}_1^A$}, \textcolor[rgb]{1,.25,0}{$\vec{\psi}_2^A$}, \textcolor[rgb]{1,.5,0}{$\vec{\psi}_3^A$}, \textcolor[rgb]{1,1,0}{$\vec{\psi}_4^A$}, \textcolor[rgb]{0,0,.95}{$\vec{\psi}_1^B$}, \textcolor[rgb]{0,.25,.90}{$\vec{\psi}_2^B$}, \textcolor[rgb]{0,.5,.85}{$\vec{\psi}_3^B$}, \textcolor[rgb]{0,1,.8}{$\vec{\psi}_4^B$}.}
\label{fig:wavefunction_notation}
\end{figure}

%\subsection{Results: disordered zig-zag edge of a graphene-like structure} \label{sec:graph_results}

\begin{figure}
% ZZ/exp_greensfunc_missingatoms.m
% Generated on Kyle's laptop
\begin{subfigure}{.5\textwidth}
  \centering
  \includegraphics[trim={5.5cm 9.5cm 5cm 8.6cm}, clip, width=\textwidth]{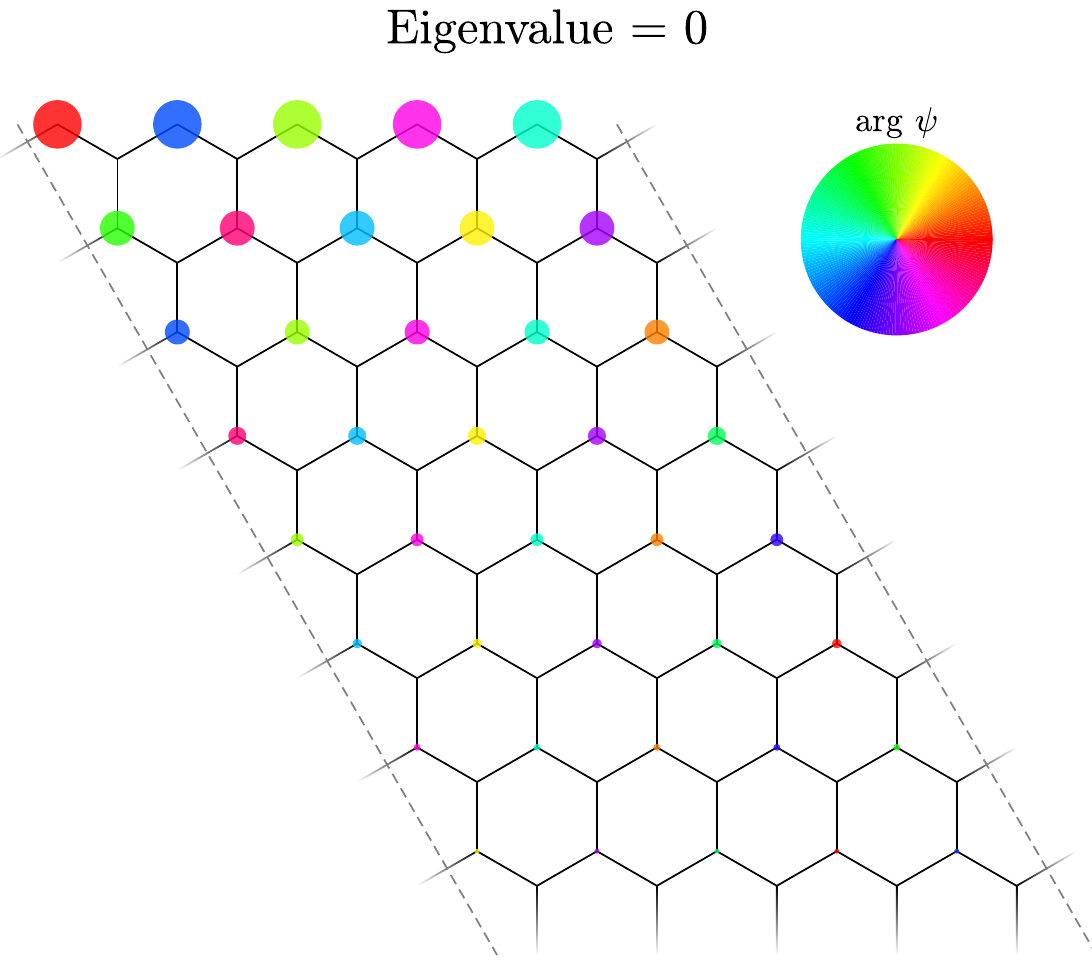}
  \caption{No defects}
  \label{fig:ZZedge_nodefects}
\end{subfigure}%
\begin{subfigure}{.5\textwidth}
  \centering
  \includegraphics[trim={5.5cm 9.5cm 5cm 8.6cm}, clip, width=\textwidth]{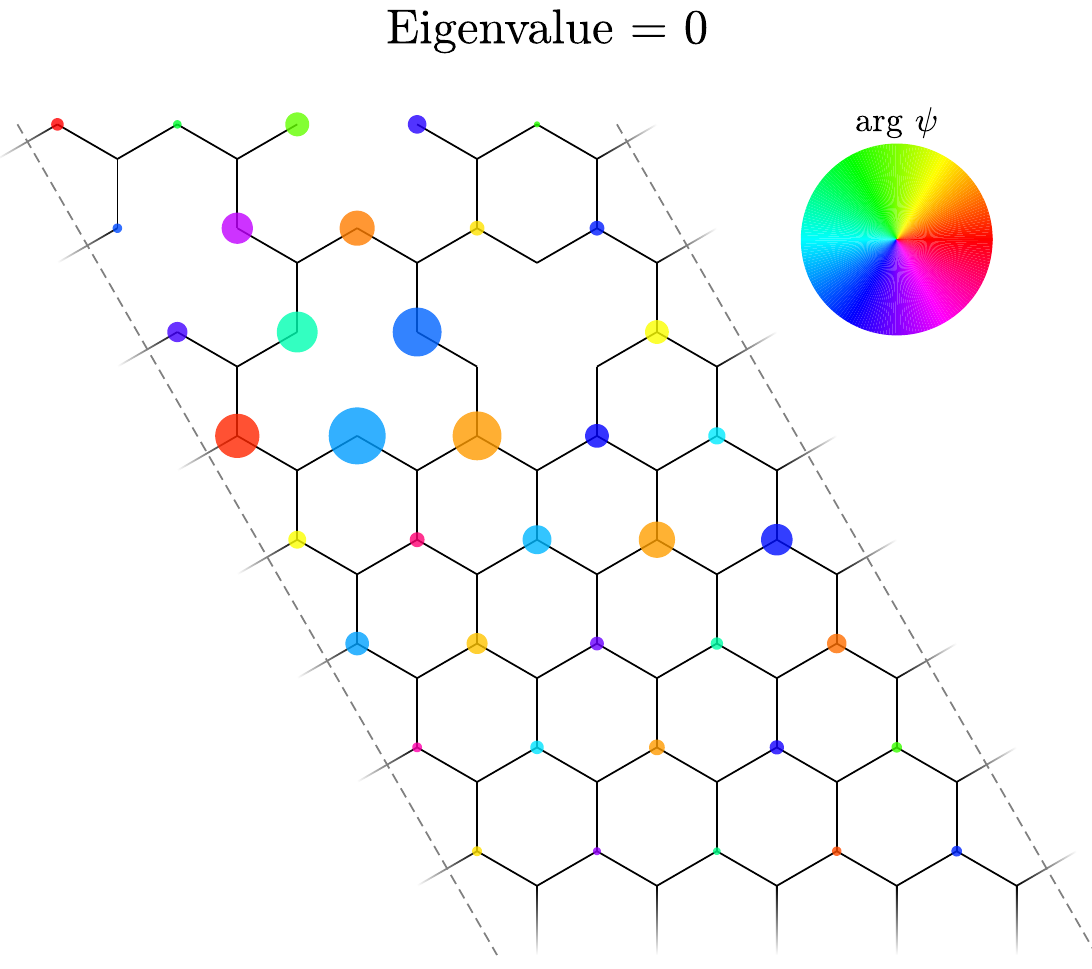}
  \caption{Atoms removed}
  \label{fig:ZZedge_removedatoms}
\end{subfigure}
\caption{Persistence of an edge state of a graphene-like structure under perturbation. In ({a}), a bound state of the zig-zag edge Hamiltonian \eqref{eq:supercell_edge_H} (edge state) when $N = 5$, $k_\parallel = \frac{1}{5} \left( - \frac{23 \pi}{30} \right)$, $V^\sigma(m,n) = 0$, $\sigma \in \{A,B\}$, and $t_j(m,n) = t$, $j \in \{0,1,2\}$ for all $m, n \in \mathbb{N} \times \mathbb{Z}$ is plotted. We find that this bound state has eigenvalue precisely $0$. The absolute value $|\psi|$ and the argument $\arg \psi$ of the wave function at each lattice point is represented by the radius and color, respectively, of the circle at that point. In ({b}), an edge state of the zig-zag edge Hamiltonian \eqref{eq:supercell_edge_H} with multiple atomic vacancies, modelled by a large onsite potential at each vacancy, is plotted. We observe that the edge state survives the perturbation: the Hamiltonian retains a bound state with eigenvalue $0$. 
%use our method to investigate the edge states of the ZZ-edge Hamiltonian under perturbations.  We consider the ZZ-edge Hamiltonian with period five in the direction parallel to the edge.  We first compute an edge state when the Hamiltonian is unperturbed.  Then, we perturb the Hamiltonian by removing four atoms from the honeycomb structure.  
Note that although the spectrum of the full two-dimensional bulk Hamiltonian is not gapped, for fixed values of $k_\parallel \notin \left\{\frac{2 \pi}{3},-\frac{2 \pi}{3}\right\}$ the essential spectrum of $\Hzzedged(k_\parallel)$ has a gap which is symmetric about $0$.}
\label{fig:ZZ_remove_atom}
\end{figure}

\subsection{Honeycomb structure tight-binding domain wall edge Hamiltonian} \label{sec:domainWall}

We now consider a second example: a ``domain wall'' interface between distinct honeycomb structures with defects. For simplicity we again assume that the interface occurs along a ``zig-zag'' edge. We make analogous assumptions to Assumptions \ref{as:eventually_nodefects}-\ref{as:edge_periodicity}. Specifically, we allow for inter-atom hopping amplitudes $t_j(m,n), j \in \{0,1,2\}$ and onsite potentials $V^\sigma(m,n), \sigma \in \{A,B\}$ which are non-trivial functions of the cell indices $m,n$ for $|m|$ sufficiently small. The restrictions we make are summed up in the following assumptions.
\begin{assumption}[Periodicity of bulk medium] \label{as:eventually_nodefects_dw}
There exist a non-negative integer $M_+$ such that for all $m \geq M_+$, $V^{\sigma}(m,n) = V^{\sigma,\infty}$ and $t_j(m,n) = t_j^\infty$ for all $n \in \mathbb{Z}$, $\sigma = A, B$, and $j \in \{0,1,2\}$ and a negative integer $M_-$ such that for all $m \leq M_-$, $V^{\sigma}(m,n) = V^{\sigma,-\infty}$ and $t_j(m,n) = t_j^{-\infty}$ for all $n \in \mathbb{Z}$, $\sigma = A, B$, and $j \in \{0,1,2\}$.
\end{assumption}
\begin{assumption}[Periodicity of disorder along edge] \label{as:edge_periodicity_dw}
There exists a positive integer $N$ such that $V^\sigma(m,n+N) = V^{\sigma}(m,n)$ and $t_j(m,n+N) = t_j(m,n)$ for all $m \in \mathbb{Z}$, $n \in \mathbb{Z}$, $j \in \{0,1,2\}$, and $\sigma = A, B$.
\end{assumption}
Under Assumptions \ref{as:eventually_nodefects_dw}-\ref{as:edge_periodicity_dw}, the problem is reduced to the study of the following Hamiltonian on a single supercell,
\begin{equation} \label{eq:supercell_edge_H_dw}
\begin{split}
	&\left[ \Hzzedged(k_\parallel) \psi \right]_{m, n} = \phantom{=} A^*(m-1,n) \psi_{m-1,n} + B^*(m,n-1) \psi_{m,n-1} + V(m,n) \psi_{m,n}  \\
        &\phantom{blablablablablablablabla} + A(m,n) \psi_{m+1,n} + B(m,n) \psi_{m,n+1}, \quad m,n \in \mathbb{Z}\times\{1,...,N\},	\\
        %- \begin{bmatrix} t_0(m,n) \psi^B_{m,n} + t_1(m-1,n) \psi^B_{m-1,n} + t_2(m,n-1) \psi^B_{m,n-1} \\ t_0(m,n) \psi^A_{m,n} + t_1(m,n) \psi^A_{m+1,n} + t_2(m,n) \psi^A_{m,n+1} \end{bmatrix} + \begin{bmatrix} V^A(m,n) \psi^A_{m,n} \\ V^B(m,n) \psi^B_{m,n} \end{bmatrix} \\ 
\end{split}
\end{equation}
where the matrices $A(m,n), B(m,n)$, and $V(m,n)$ are defined by \eqref{eq:hop_matrices}, the $V^\sigma(m,n)$, $\sigma \in \{A,B\}$ and $t_j(m,n)$, $j \in \{0,1,2\}$ are as in Assumptions \ref{as:eventually_nodefects_dw} and \ref{as:edge_periodicity_dw} and
\begin{equation}
%	&t_2(m,N+1) = t_2(m,1) \text{ for } m \in \mathbb{Z} \text{ and } n \in \{1,...,N\}	\\
%	&V^{\sigma}(m,n) = V^{\sigma,\infty} \text{ and } t_j(m,n) = t_j^\infty \text{ for } m \in \mathbb{Z}, m > M_+, n \in \{1,...,N\}, j \in \{0,1,2\}, \sigma \in \{A,B\} \\
%	&V^{\sigma}(m,n) = V^{\sigma,-\infty} \text{ and } t_j(m,n) = t_j^{-\infty} \text{ for } m \in \mathbb{Z}, m < M_-, n \in \{1,...,N\}, j \in \{0,1,2\}, \sigma \in \{A,B\} \\
	\psi_{0,n} = 0 \text{ for } n \in \{1,...,N\} \quad \psi_{m,0} = e^{- i k_\parallel N} \psi_{m,N} \quad \psi_{m,N+1} = e^{i k_\parallel N} \psi_{m,1} \text{ for } m \in \mathbb{Z}.
\end{equation}
Under Assumptions \ref{as:eventually_nodefects_dw} and \ref{as:edge_periodicity_dw}, $\Hzzedge(k_\parallel)$ is an infinite eventually periodic Hamiltonian in the sense of Definition \ref{def:ev_per_2} and hence amenable to our method.

In Figures \ref{fig:DW_dispersion} and \ref{fig:DW_missingatoms} we present results of numerical computations of spectral data of the Hamiltonian \eqref{eq:supercell_edge_H_dw} for the case
\begin{equation} \label{eq:interesting_case}
	V^{\sigma,\pm \infty} = 0, \sigma \in \{A,B\}, \text{ and } t_0^{-\infty} = t_1^{-\infty}, \text{ } \frac{ |t_2^{-\infty}| }{ |t_1^{-\infty}| } > 1, \text{  and  } t_0^\infty = t_1^\infty, \text{ } \frac{ |t_2^\infty| }{ |t_1^\infty| } < 1.
\end{equation}
This case, studied for example in \cite{2000Chamon}, is of particular interest to us because hard truncation yields inaccurate results in this case (see Lee-Thorp \cite{2016Lee-Thorp}, in particular Figure 26.7). In Figure \ref{fig:DW_dispersion} we demonstrate this by comparing the edge dispersion curves computed with hard truncation and the Green's function method. In Figure \ref{fig:DW_missingatoms} we study a bound state of \eqref{eq:supercell_edge_H_dw} in the presence of defects. 

%\subsection{Results: zig-zag domain wall edge in a dimerized honeycomb structure} \label{sec:dw_results}

\begin{figure}
\centering
\begin{subfigure}[t]{.45\textwidth}
\includegraphics[width=\textwidth]{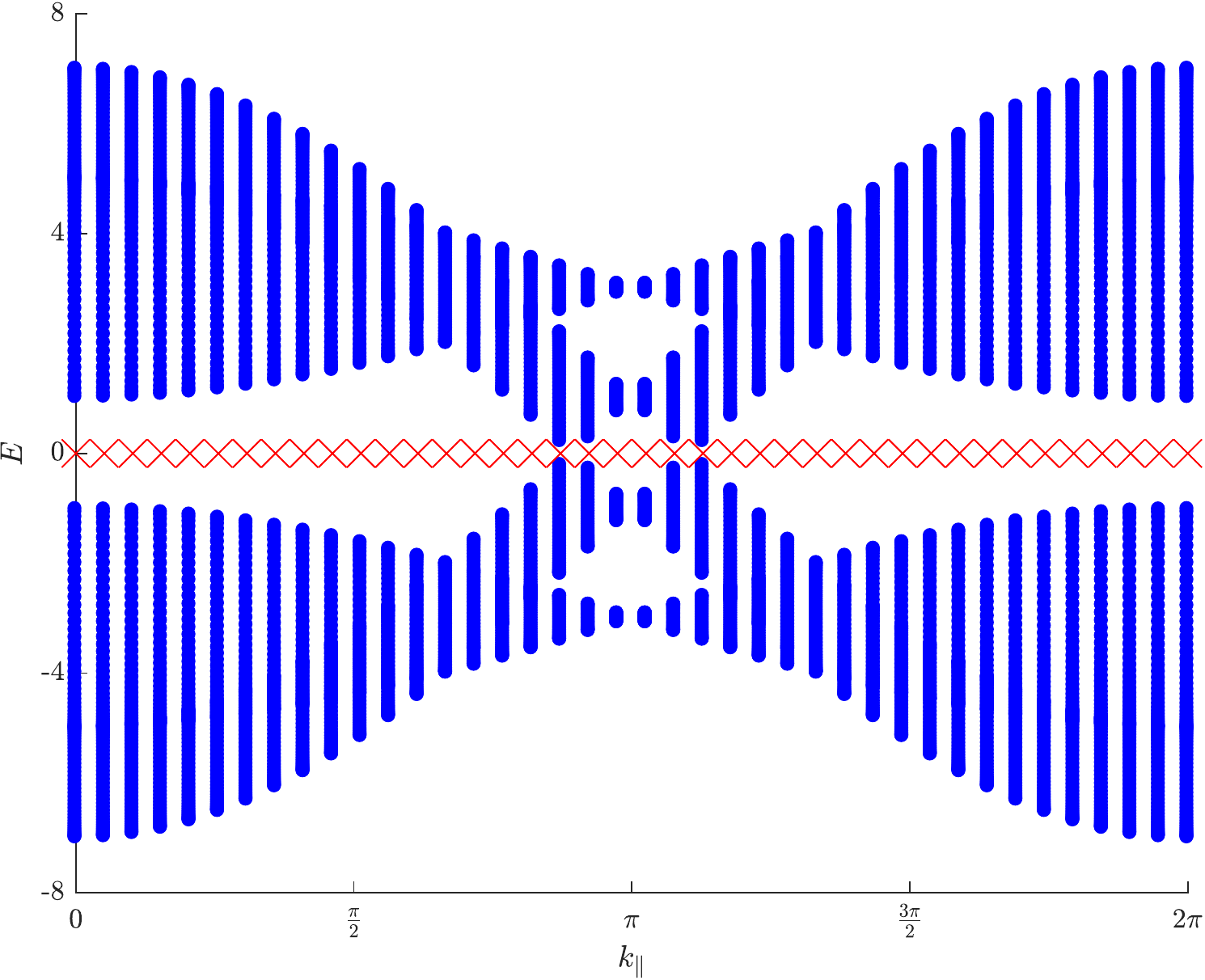}
\caption{}
\label{fig:HT_dispersion}
\end{subfigure} \quad %
\begin{subfigure}[t]{.45\textwidth}
\includegraphics[width=\textwidth]{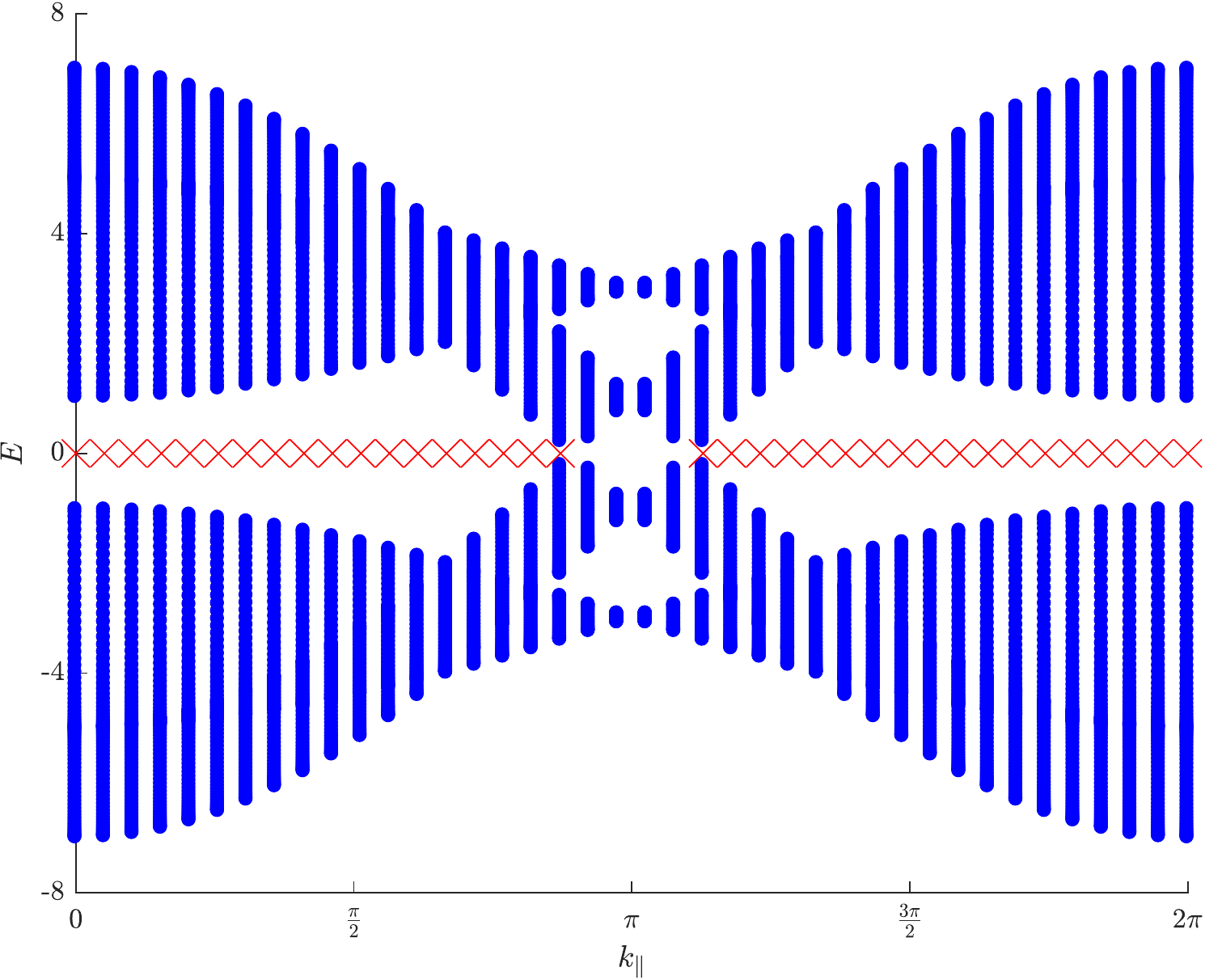}
\caption{}
\label{fig:GF_dispersion}
\end{subfigure}
\caption{({a}) Dispersion curves $E(k_\parallel)$ of the edge Hamiltonian \eqref{eq:supercell_edge_H_dw} computed with the hard truncation method, and ({b}) dispersion curves of the edge Hamiltonian \eqref{eq:supercell_edge_H_dw} computed using the Green's function method. Zero eigenvalues are plotted with red $\times$s, while all other eigenvalues are plotted in blue. For $k_\parallel$ in a neighborhood of $\pi$, hard truncation of the Hamiltonian yields spurious zero modes which are eliminated when the proper boundary conditions are applied.}
\label{fig:DW_dispersion}
\end{figure}

\begin{figure}
\centering
\begin{subfigure}[t]{.49\textwidth}
\includegraphics[trim={5cm 8cm 4.25cm 7.5cm}, clip, width=\textwidth]{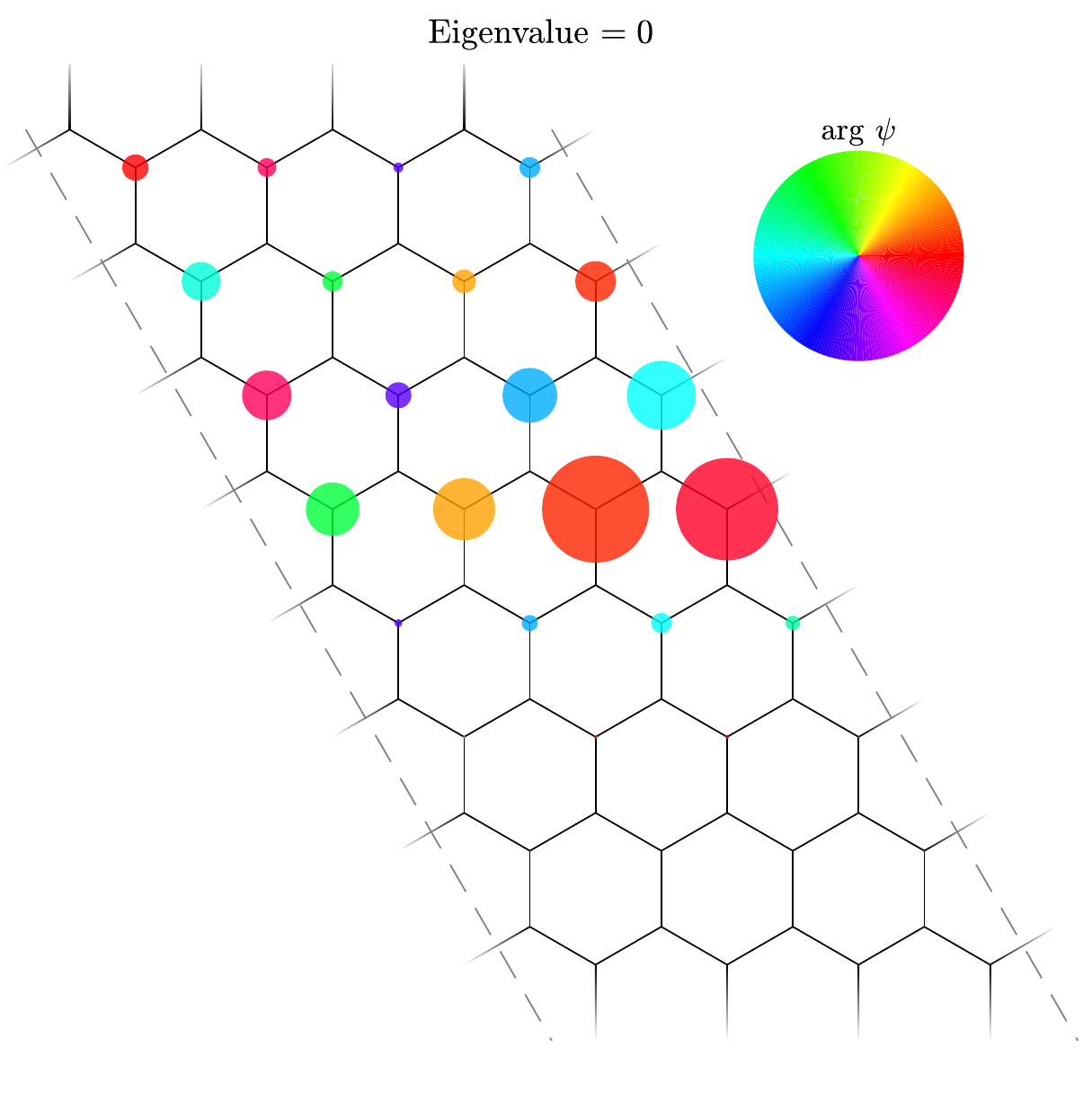}
\caption{No defects}
\label{fig:DW_SC_1}
\end{subfigure} %\quad %
\begin{subfigure}[t]{.49\textwidth}
\includegraphics[trim={5cm 8cm 4.25cm 7.5cm}, clip, width=\textwidth]{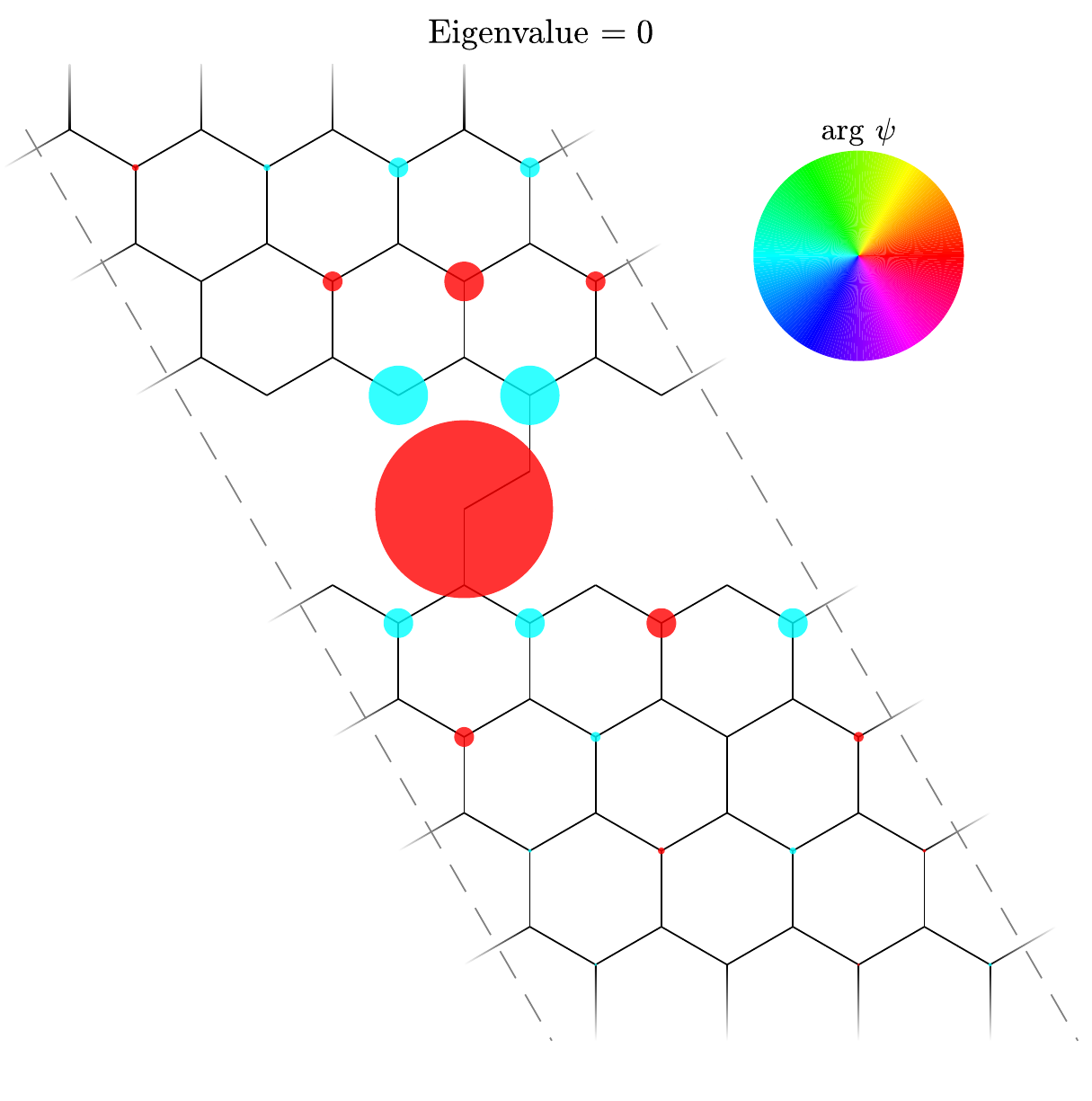}
\caption{Atoms removed}
\label{fig:DW_SC_2}
\end{subfigure}
%\includegraphics[width=\textwidth]{figures/DW_supercell_edgestate_2}
%\begin{subfigure}[t]{.45\textwidth}
%\includegraphics[width=\textwidth]{figures/DW_supercell_edgestate_3}
%\caption{}
%\label{fig:DW_SC_3}
%\end{subfigure} \quad %
%\begin{subfigure}[t]{.45\textwidth}
%\includegraphics[width=\textwidth]{figures/DW_supercell_edgestate_4}
%\caption{}
%\label{fig:DW_SC_4}
%\end{subfigure}
\caption{({a}) Plot of one of the four-fold degenerate zero eigenvalue bound states of the Hamiltonian \eqref{eq:supercell_edge_H_dw} when $k_\parallel = \frac{\pi}{4}$ when the bulk structures far from the domain wall edge are defined by \eqref{eq:interesting_case} without disorder. The absolute value $|\psi|$ and argument $\arg \psi$ of the wavefunction at each point on the lattice is represented by the radius and color, respectively, of the circle at each point. ({b}) An edge state of the Hamiltonian \eqref{eq:supercell_edge_H_dw} when $k_\parallel = \frac{ \pi }{ 4 }$ under Assumption \eqref{eq:interesting_case} with atomic vacancies at the edge modeled by large onsite potentials. We observe that in this case the Hamiltonian has a single non-degenerate bound state with eigenvalue $0$.} 
\label{fig:DW_missingatoms}
\end{figure}

\subsection{Honeycomb structure continuum edge Hamiltonian} \label{sec:continuum} 
%\jl{I suggest change the terminology "discretized continuum" (which is really weird)} \aw{OK} 

\begin{figure}
\centering
\begin{subfigure}{.5\textwidth}
  \centering
  \includegraphics[trim = 3in 0in 3in 0in, clip, scale=.25]{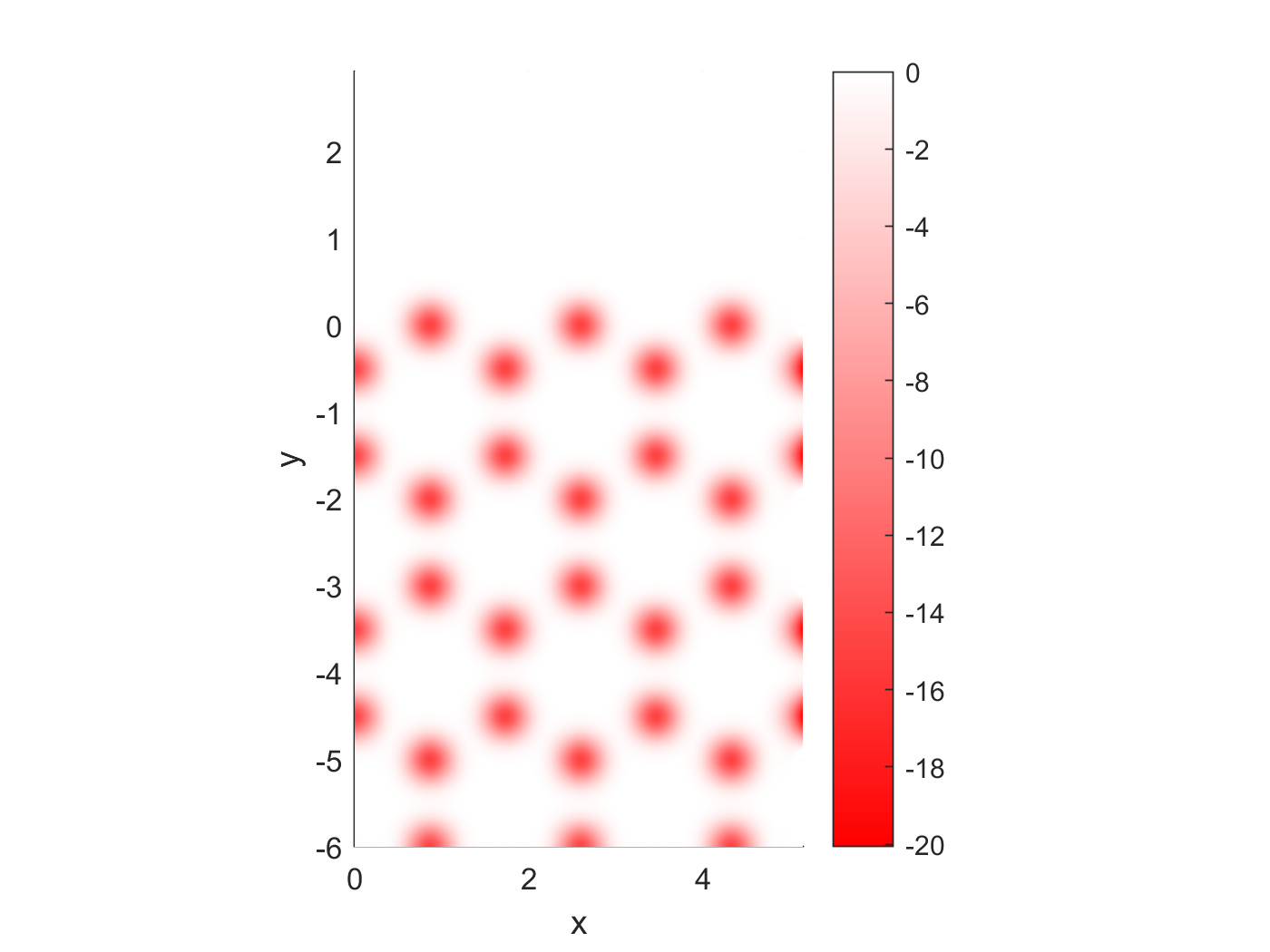}
  \caption{Potential, no defects.}
  \label{fig:cont_ext_nodefects}
\end{subfigure}%
\begin{subfigure}{.5\textwidth}
  \centering
  \includegraphics[trim = 3.5in 0in 3.5in 0in, clip, scale=.25]{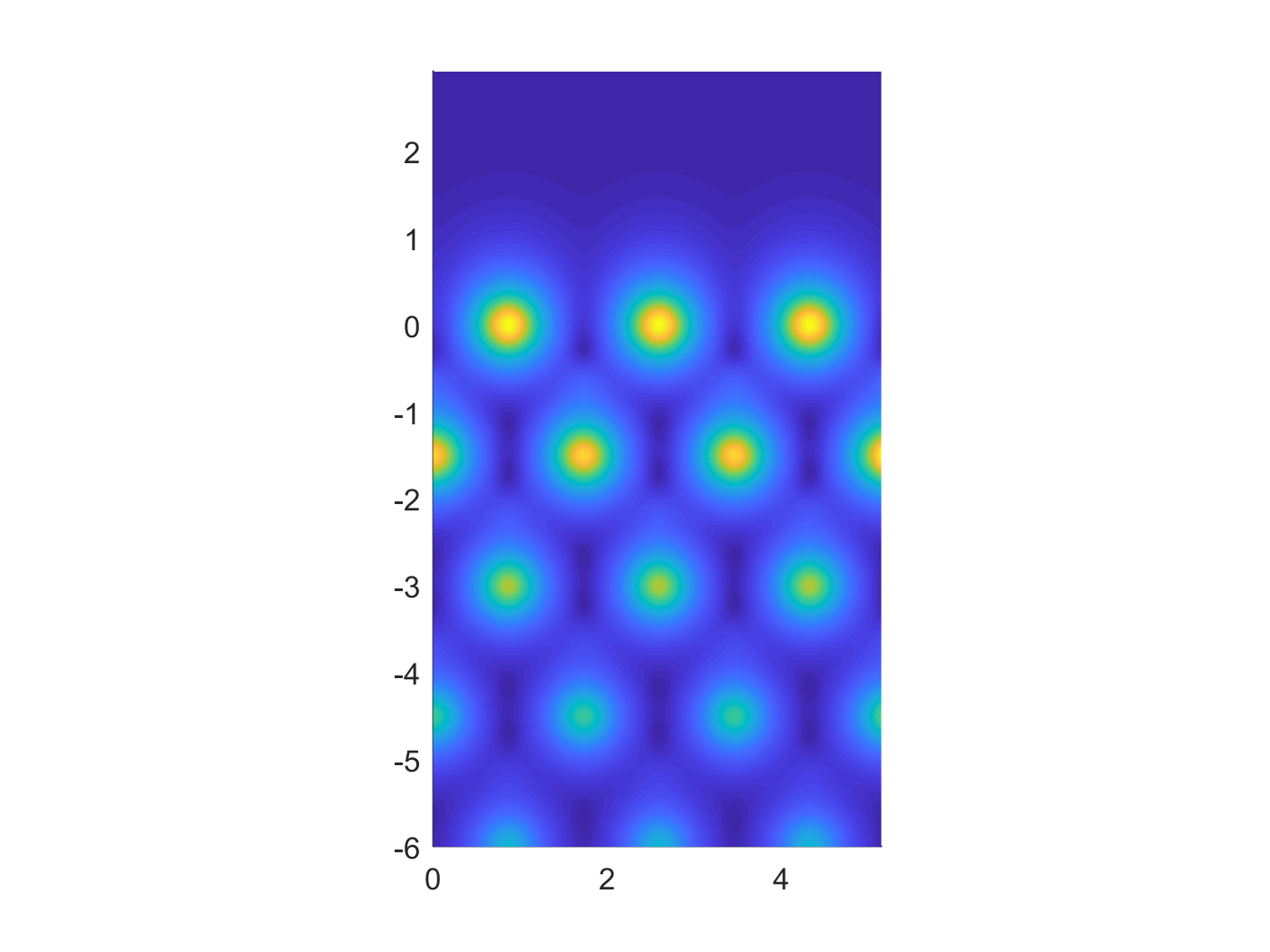}
  \caption{Edge state, no defects.}
  \label{fig:cont_wf_nodefects}
\end{subfigure}
\begin{subfigure}{.5\textwidth}
  \centering
  \includegraphics[trim = 3in 0in 3in 0in, clip, scale=.25]{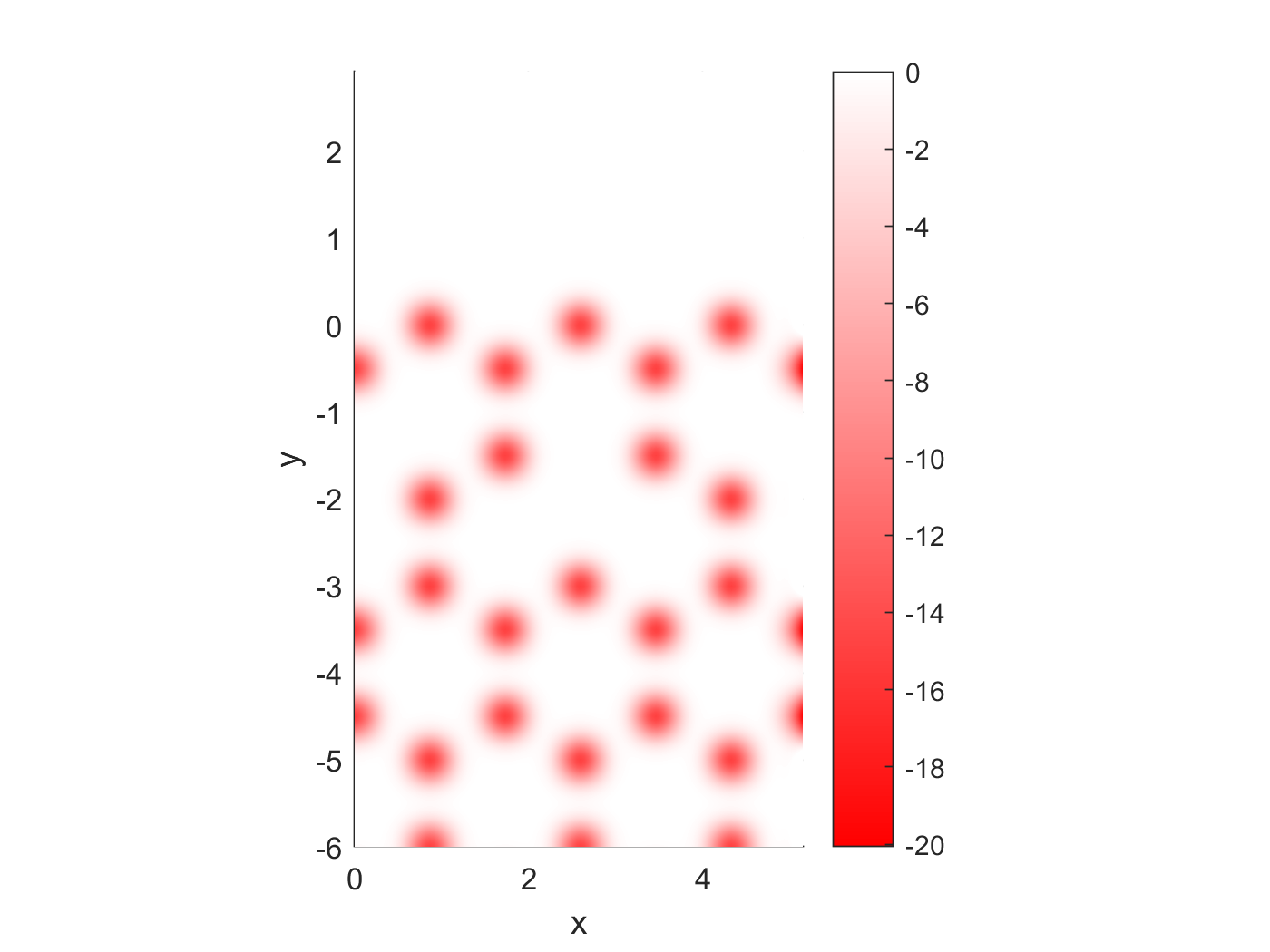}
  \caption{Potential, 2 atoms missing.}
  \label{fig:cont_ext_defects}
\end{subfigure}%
\begin{subfigure}{.5\textwidth}
  \centering
  \includegraphics[trim = 3.5in 0in 3.5in 0in, clip, scale=.25]{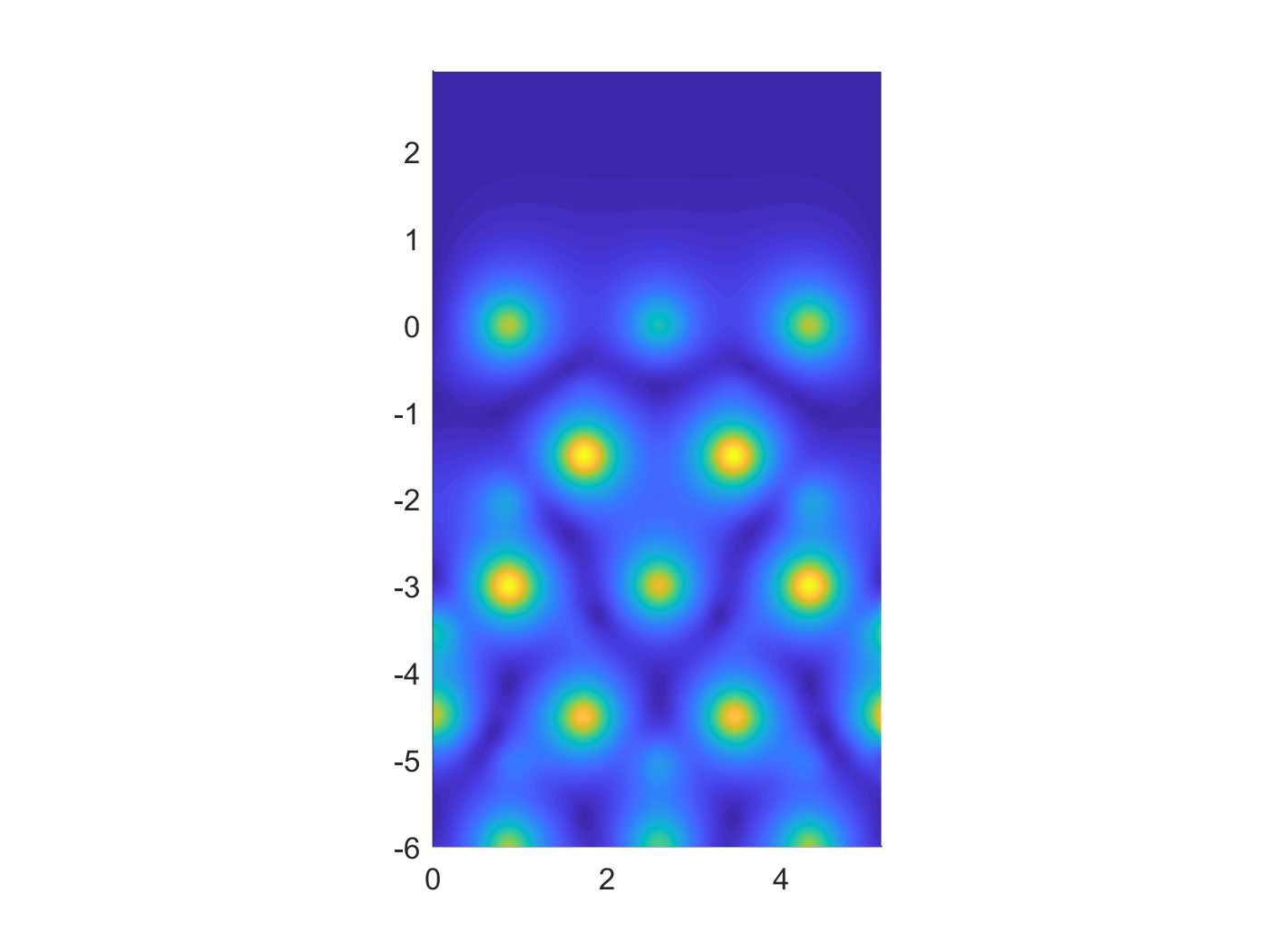}
  \caption{Edge state, 2 atoms missing.}
  \label{fig:cont_wf_defects}
\end{subfigure}
\caption{In figures (a) and (b), we plot the potential and corresponding edge state of \eqref{eq:hcedge} with $k_\parallel = \frac{ 3.8 \pi }{ 3 }$ when the edge does not have defects.  Note that we choose co-ordinates such that the edge is at $y = 0$.  To generate figures (c) and (d), we have removed two atoms as can be seen in the plot of the external potential (c). An edge state with $k_\parallel = \frac{3.8 \pi}{3}$ is also plotted. The edge state plotted in (b) has eigenvalue -1.382, while the edge state plotted in (d) has eigenvalue -1.411.}
\label{fig:cts_figure}
\end{figure}

\begin{figure}
    \centering
    \includegraphics[scale = .3]{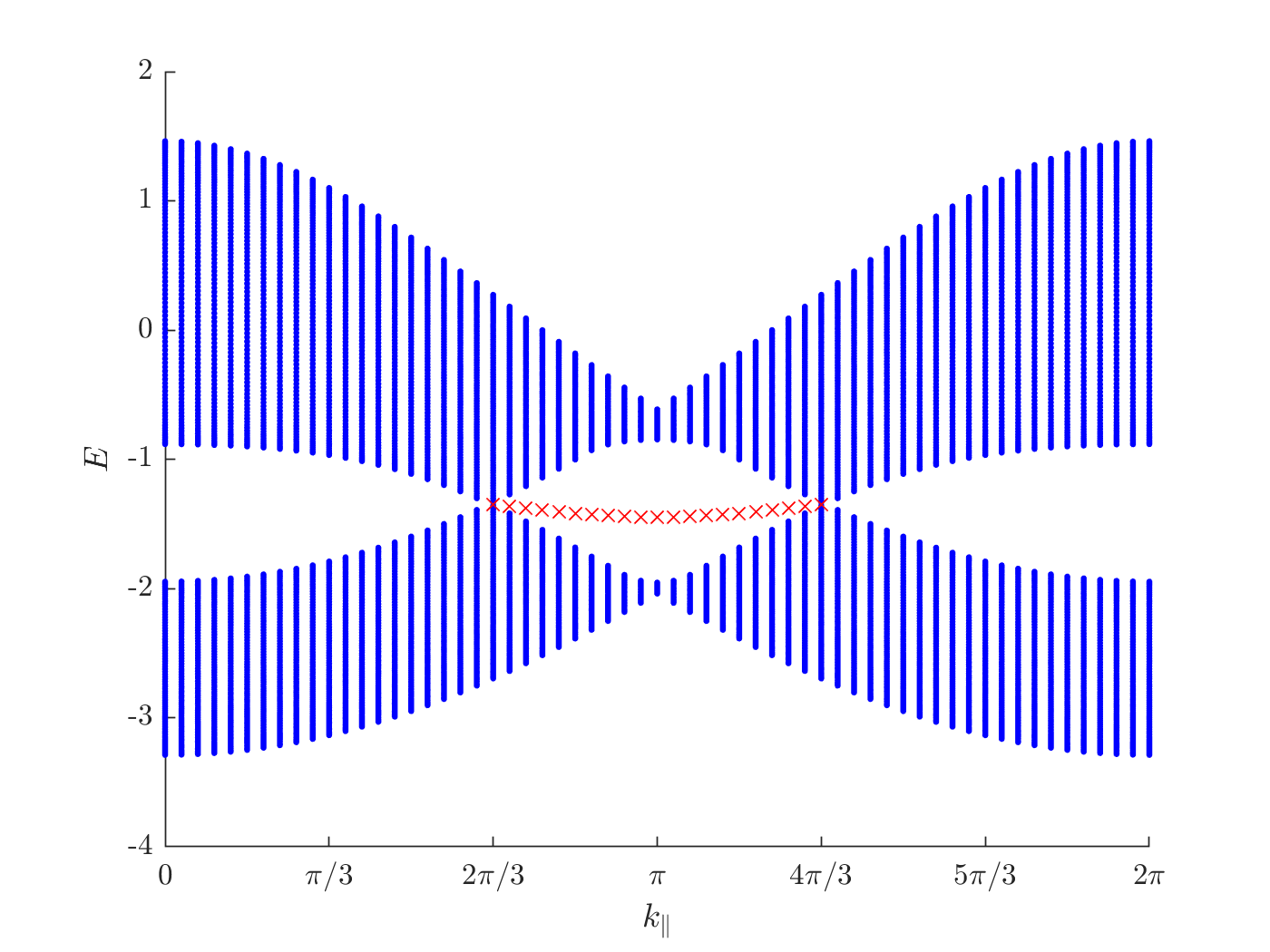}
    \caption{Dispersion curves $E(k_\parallel)$ of the edge Hamiltonian \eqref{eq:hcedge} without defects. Essential spectrum (computed via two-dimensional Bloch theory) is shown in blue while edge states computed via the Green's function method are marked with red $\times$s. Note that the curve of edge state eigenvalues is not flat, in contrast to the prediction of the analogous tight-binding model.}
    \label{fig:cont_dispersion_curve}
\end{figure}

Our method can also be used to compute spectral data of continuum Schr\"odinger edge models discretized by a finite difference scheme. 
Following \cite{2018FeffermanWeinstein}, we consider the eigenvalue problem $\Hcedge \psi = E \psi$ for the Schr\"odinger operator
\begin{equation} \label{eq:hcedge}
    \Hcedge := - \Delta + V({x})
\end{equation}
acting on $L^2(\mathbb{R}^2)$, where the real potential function $V(x)$ models a ``half-plane'' honeycomb structure. More precisely, introduce the lattice vectors
\begin{equation}
    v_1 = \frac{1}{2}\begin{bmatrix} 3 \\ \sqrt{3} \end{bmatrix}, \quad v_2 = \frac{1}{2} \begin{bmatrix} 3 \\ - \sqrt{3} \end{bmatrix}.
\end{equation}
Then we assume
\begin{equation} \label{eq:sum_V}
    V({x}) = \sum_{{v} \in \Lambda_+} V_0({x} - {v})
\end{equation}
where $\Lambda_+$ is the ``half-plane'' honeycomb lattice $\Lambda_+ = \left\{ m v_1 + n v_2 : m \in \mathbb{Z}, n \in \mathbb{N} \right\}$. For simplicity we take the potential at each lattice vertex $V_0({x})$ %\jl{why switch to $z$ here?} \aw{good point, fixed}
to be a Gaussian well (see Figure \ref{fig:cts_figure}). Such a Hamiltonian models electronic states at a ``zig-zag'' edge of a graphene-like structure.

We model atomic vacancy defects at the edge by omitting terms in the sum \eqref{eq:sum_V}. Under assumptions analogous to Assumptions \ref{as:eventually_nodefects_dw} and \ref{as:edge_periodicity_dw} we obtain a Hamiltonian which is periodic under translation by $N v_1$ (parallel to the edge) for sufficiently large $N > 1$ and eventually periodic (see Remark \ref{rem:not_per}) with respect to translation by $v_2$ for $x$ sufficiently far from the edge, both into the bulk of the material and into free space in the opposite direction. Imposing Bloch periodicity in the edge direction $\Phi(x + N v_1;k_\parallel) = e^{i N k_\parallel} \Phi(x;k_\parallel)$, we derive a family of eigenvalue problems on $L^2([0,N] \times \mathbb{R})$ %\jl{$[0, N]$ or $[1, N]$?} \aw{you're right, it should start at $0$, fixed}
parameterized by $k_\parallel$. Discretizing any one of these eigenvalue problems by a five-point stencil we obtain an infinite eventually periodic Hamiltonian in the sense of Definition \ref{def:ev_per_2} whose bound states can be computed using the Green's function method. Our results are shown in Figures \ref{fig:cts_figure} and \ref{fig:cont_dispersion_curve}.

\begin{remark} \label{rem:not_per}
In fact the Hamiltonian cannot be exactly eventually periodic even in the absence of defects because of the exponential tail of the potential at each vertex. Since the error accumulated by assuming that it is in fact eventually periodic is exponentially small (indeed can be made arbitrarily small by truncating further away from any defects) we ignore this issue. 
\end{remark}

\section{Conclusion} In this work we have introduced a numerical method for computing spectral data of a class of discrete Hamiltonians which model edge states: electronic states localized at edges of one and two-dimensional materials. Such states have generated considerable interest because of their potential as robust wave-guides. Our method is exact up to rounding error for discrete models and exact up to discretization error for continuum models and is hence a significant improvement on the na\"ive ``hard truncation'' method which may yield spurious states localized at the truncation rather than the edge of the structure. We have verified the effectiveness of our method through studies of the robustness of edge states to defects in various one and two-dimensional models.

Much remains to be explored related with edge states and their stability. For instance, it should be possible using the ideas presented in this paper to realize a scheme for studying the propagation in time of wave-packets propagating along edges which doesn't require ``hard truncation'' into the bulk. Another interesting direction would be to adapt our method to compute resonances resulting from defects (localized in both directions as well as along one-dimensional ``edges'') in periodic media. Locating such resonances (specifically, computing their imaginary part) would give information about the life-times of approximate bound states created by such defects. Finally, our methods can already be directly applied to the computation of ``surface'' states of three dimensional materials in the presence of defects under analogous assumptions to Assumptions \ref{as:eventually_nodefects} and \ref{as:edge_periodicity}.

\begin{appendices}

\section{Proof of existence of boundary conditions for arbitrary finite range eventually periodic Hamiltonians} \label{sec:method_in_general}

In this section we discuss a more general framework for proving the existence of appropriate boundary conditions (i.e. those which ensure that the Green's function we compute maps $l^2(\mathbb{N};\mathbb{C}^N)$ to itself) for computing the solution of \eqref{eq:G_generall}.

The only assumptions we need on the Hamiltonian $H \in \mathbb{C}^{\mathbb{N} \times \mathbb{N}}$ are made explicit next.
\begin{assumption}[$H$ is banded]
\label{a:generalH_banded}
There exists a positive integer $B$ such that $H_{m,n} = 0$ if $|m-n| > B$.
\end{assumption}
\begin{assumption}[Periodicity of the bulk]
\label{a:generalH_eventually_periodic}
There exists a period $p$ and a non-negative integer $M_+$ such that for all $m \geq M_+$, $H_{m,n} = H_{m+p,n+p}$.
\end{assumption}
Note that we do not even require that $H$ is Hermitian; our method works for non-Hermitian operators as well.  However, if $H$ is not Hermitian, then Theorem \ref{th:Riesz} is no longer guaranteed.  Unfortunately, the conclusion of Theorem \ref{th:Riesz} does not hold for general operators, so one would have to be more careful in justifying that step in the non-Hermitian case.

Let us define some notation.  Under Assumptions \ref{a:generalH_banded} and \ref{a:generalH_eventually_periodic}, the matrix $z-H$ can be split up as shown next.
\begin{equation}
z-H = \left[\begin{array}{c}
(z-H)_\fin \\\hline
(z-H)_\per
\end{array}\right] = \left[\begin{array}{p{.3cm}p{.3cm}p{.3cm}p{.3cm}p{.3cm}p{.3cm}p{.3cm}p{.4cm}}
& & & \multicolumn{1}{c|}{} & & & & \\
\multicolumn{4}{c|}{(z-H)_\text{defect}} &  \\\cline{5-5}
\multicolumn{4}{c|}{} & \multicolumn{1}{c|}{\mathcal{P}}\\\hline
 & & & \multicolumn{3}{|c|}{P} \\\cline{4-7}
 & & & & \multicolumn{3}{|c|}{P} \\\cline{5-8}
 & & & & & \multicolumn{3}{|c|}{P} \\\cline{6-8}
 & & & & & & & $\ddots$
\end{array}\right]. \label{zH_block_structure}
\end{equation}
We call $(z-H)_\per$ the periodic-part of $z-H$; it consists of the matrix $P$ (representing one period) periodically repeating off to infinity.  We call $(z-H)_\text{fin}$ the finite-part; it consists of the defect-part $(z-H)_\text{defect}$ (which is where defects are allowed) and the last portion of a $P$ matrix, which we have denoted $\mathcal{P}$.  We assume that $P$ has no ``extra'' columns of zeros on either the left or right end: It has a nonzero entry in its leftmost and rightmost column.  We now introduce some more notation by defining the sizes of these matrices:
\begin{align*}
(z-H)_\text{fin} &\in \mathbb{C}^{N_d \times \mathbb{N}}, & (z-H)_\per &\in \mathbb{C}^{\mathbb{N} \times \mathbb{N}}, & (z-H)_\text{defect} &\in \mathbb{C}^{N_d \times N_d}, \\
P &\in \mathbb{C}^{p \times w}, & \mathcal{P} &\in \mathbb{C}^{p \times b_+}.
\end{align*}
We finally define one more piece of notation: The first $P$ overlaps with $b_-$ columns of $(z-H)_\text{defect}$.  That is, the leftmost nonzero entry of $(z-H)_\per$ occurs in column $N_d - b_- + 1$.

Now that we have defined notation, let's look at the problem we actually want to solve.  We wish to compute the upper-left $N_d \times N_d$ block of $G(z)$; that is, we wish to compute the first $N_d$ entries of the solution to 
\begin{equation}
(z-H)g_j = e_j
\end{equation}
for $j = 1, 2, ..., N_d$.  First, we remark that this solution is unique in $\ell^2(\mathbb{N})$.

\begin{lemma}
\label{lem:uniqueness_of_g}
If $z$ is not in the spectrum of $H$, then the linear system $(z-H) g_j = e_j$ has a unique solution in $\ell^2(\mathbb{N})$.
\end{lemma}
\begin{proof}
Under Assumptions \ref{a:generalH_banded} and \ref{a:generalH_eventually_periodic}, it is straightforward to show that $H$ is a bounded operator on $\ell^2(\mathbb{N})$.  The conclusion then follows.
\end{proof}

Our goal is to compute the finite-part of $g_j$.  In order to do this, there are three conditions $g_j$ must to satisfy:
\begin{enumerate}
\item $g_j$ must satisfy all the finite-part equations: $(z-H)_\fin g_j = e_j$.
\item $g_j$ must satisfy all the periodic-part equations: $(z-H)_\per g_j = 0$.
\item $g_j$ must be in $\ell^2$.
\end{enumerate}
By Lemma \ref{lem:uniqueness_of_g}, $g_j$ is the \emph{only} vector which satisfies all three of these conditions.

We can reduce these three conditions down to just two by combining Conditions 2 and 3: We see that $g_j$ satisfies both conditions if and only if $g_j$ is in the nullspace of $(z-H)_\per$, considered as an operator on $\ell^2(\mathbb{N})$.  This leaves us with just two requirements:
\begin{enumerate}[1'.]
\item $g_j$ must satisfy all the finite-part equations: $(z-H)_\fin g_j = e_j$.
\item $g_j$ must be in the nullspace of $(z-H)_\per$, considered as an operator on $\ell^2(\mathbb{N})$.  
\end{enumerate}
Note that these requirements are again both necessary and sufficient: A vector satisfies these two conditions if and only if it is $g_j$.

The rest of this section is organized as follows.  First, we characterize the nullspace of $(z-H)_\per$, considered as an operator on $\ell^2(\mathbb{N})$.  Once we have this nullspace, we show how to construct the boundary conditions.  Lastly, we prove that this procedure for incorporating the boundary conditions always results in a system whose unique solution is the finite-part of $g_j$.

To compute the nullspace of $(z-H)_\per$, we proceed in a few steps.  First, Lemma \ref{lem:eigenpairs_of_S_p} characterizes the eigenspaces of the operator that shifts a vector by one period, $S_p: [\psi_1, \psi_2, ... ]^\trans \mapsto [\psi_{p+1}, \psi_{p+2}, ...]^\trans$.  Second, using the fact that $(z-H)_\per$ and $S_p$ commute, we show in Lemma~\ref{lem:nullspace_zH_per} that the nullspace of $(z-H)_\per$ is spanned by a finite set of (generalized) eigenvectors of $S_p$, which fall into two particular types.  (Note that the standard ``commuting operators have a common set of eigenvectors'' theorem does not apply here because neither $(z-H)_\per$ nor $S_p$ is known \emph{a priori} to be diagonalizable; neither operator is normal.)  Using these particular forms for the basis vectors of $\nul(z-H)_\per$, we show how a polynomial eigenvalue problem can be used to actually compute such a basis.
\begin{lemma}
\label{lem:eigenpairs_of_S_p}
\emph{(Characterization of the eigenspaces of $S_p$, as an operator on $\ell^2(\mathbb{N})$)}\\
Let $\eta \in (-1,0) \cup (0, 1)$ and $\xi \in \mathbb{C}^p$, $\xi \ne 0$.  Then, 
\begin{equation}
\psi = [\begin{array}{cccc}
\xi, & \eta\xi, & \eta^2\xi, & ...
\end{array}]^\trans, \label{evecForm}
\end{equation}
is an eigenvector of $S_p$ corresponding to eigenvalue $\eta$.
This characterizes the nonzero eigenspaces.  The generalized eigenspace of the zero-eigenvalue is the space of vectors in $\mathbb{C}^\mathbb{N}$ with only a finite number of nonzero entries.
\end{lemma}
\begin{proof}
The nonzero claim is straightforward.  For the zero-eigenspace, note that $\psi$ is a generalized eigenvector if and only if $S_p^k \psi = 0$ for some $k$.  The $\psi$ which satisfy this are precisely the ones with only a finite number of nonzero entries.
\end{proof}

\begin{lemma}
\label{lem:nullspace_zH_per}
Let $N_b$ be a positive integer (to be chosen later).  If $z$ is not in the spectrum of $H$, then the nullspace of $(z-H)_\per$ is spanned by a set of $N_d$ vectors which fall into one of the following categories:
\begin{enumerate}[Type 1:, leftmargin=\widthof{Type }+2\parindent]
\item[Type 1a:] Vectors that are only nonzero in the first $N_b$ entries;
\item[Type 1b:] Vectors that have a finite number of nonzero entries, but are zero in the first $N_b$ entries;
\item[Type 2:] Vectors of the form \eqref{evecForm} with $\eta \in (-1,0) \cup (0, 1)$.
\end{enumerate}
\end{lemma}
\begin{proof}
See Appendix \ref{sec:proof_of_lemma_nullspace}.
\end{proof}

Next, we show how to find such a basis for $\nul(z-H)_\per$ by solving a generalized eigenvalue problem.  First, let us define two submatrices $R$ and $S$ of $(z-H)_\per$ as follows.  $R$ and $S$ are square matrices of identical size and have side lengths that are equal to an integer multiple of $p$.  $R$ and $S$ are next to each other (with $S$ immediately to the right of $R$), and both are aligned with the top of $(z-H)_\per$.  Finally, we must choose $R$ and $S$ large enough so that they contain all the nonzero entries of $(z-H)_\per$ in the rows they cover.  Additionally, $R$ must be chosen so that it contains all the all the nonzero entries of Type 1b vectors.  With these definitions, the basis of Lemma \ref{lem:nullspace_zH_per} can be found by solving the generalized eigenvalue problem $R + \lambda S$.  In particular, the eigenvectors of $R + \lambda S$ can be extended (in the sense of the following definition) to basis vectors of $\nul(z-H)_\per$.

\begin{remark}
It is not always known a priori how large $R$ should be chosen.  But if a choice fails (i.e., you do not get enough eigenvectors of $R + \lambda S$), then you simply try again with a larger choice.
\end{remark}

\begin{definition}
Let $x$ be an eigenvector of $R + \lambda S$; that is, for some $\lambda_0$, $(R + \lambda_0 S)x = 0$.  Then, we define the \textbf{\em extension} of $x$ to $\ell^2(\mathbb{N})$ as follows.  If $\lambda_0 = 0$, then the extension is $\xi = [\begin{array}{ccc}
0_{N_b} & x & 0_{\infty}
\end{array}]^\trans$.  If $\lambda_0 \ne 0$, then we define the $N_b+1$ through $N_b + N_R$ entries to be $x$; i.e., in MATLAB notation, $\xi_{N_b+1:N_b+N_R} = x$.  The next $N_R$ entries are defined to be $\lambda_0 x$, then $\lambda_0^2 x$, and so on.  The first $N_b$ entries of $\xi$ are defined by continuing this pattern backwards.  Conversely, we define the \textbf{\em restriction} $v_\textnormal{restr}$ of a vector $v \in \ell^2(\mathbb{N})$ to be $v_\textnormal{restr} = v_{N_b+1:N_b+N_R}$.
\end{definition}

To see that it is enough to solve the generalized eigenproblem $R+\lambda S$, we first note that, since we chose the sides of $R$ to be an integer multiple of $p$ and chose them large enough, restrictions of the Type 1b and Type 2 basis vectors form solutions of $(R + \lambda S)\xi = 0$, corresponding to $|\lambda| < 1$.  Second, we show that these are the \emph{only} such solutions.

\begin{lemma}
\label{lem:transferMatrix}
There are exactly $N_d - N_b$ eigenvectors corresponding to eigenvalues of $R + \lambda S$ with $|\lambda| < 1$.  Combining extensions of these eigenvectors with the vectors $\{e_j\}_{j=1}^{N_b}$ yields a basis for $\nul(z-H)_\per$ of the form in Lemma \ref{lem:nullspace_zH_per}.
\end{lemma}
\begin{proof}
See Appendix \ref{sec:proof_transferMatrix}.
\end{proof}

We note that there could also be generalized eigenvectors of $R+\lambda S$.  However, we do not want these \dash they correspond to generalized zero-eigenvectors of $(z-H)_\per$.

\begin{remark}
Note that the restriction that we choose $R$ large enough to encompass all the nonzero entries of the finitely-many nonzero entry basis vectors is important; it guarantees that the Type 1b vectors will be zero-eigenvectors of $R+\lambda S$.  An example of why we might have to choose $R$ to be large is illustrated with the following example.  Suppose
\begin{equation}
P = \left[\begin{array}{cccccccc}
1 & 0 & 0 & 0 & 0 & 1 & 0 & 0 \\
0 & 1 & 0 & 0 & 0 & 0 & 1 & 0 \\
0 & 0 & 1 & 0 & 0 & 0 & 0 & 1 \\
0 & 0 & 0 & 0 & 1 & 0 & 0 & 0
\end{array}\right].
\end{equation}
Then, there is a Type 1 vector given by
\begin{equation}
[0_{N_b},1,0,0,0,0,-1,0,0,0,0,1,0,0,0,-1, 0_{\infty}]^\trans.
\end{equation}
This vector has nonzero values past the end of $\mathcal{P}$.  We note that this example can easily be generalized to make the last nonzero value go arbitrarily far past the end of $\mathcal{P}$ by choosing $P$ to be
\begin{equation}
P = \left[\begin{array}{cccccccc}
 & & \multicolumn{1}{c|}{} & 0 & 0 & \multicolumn{1}{|c}{} \\
 & I_k & \multicolumn{1}{c|}{} & \vdots & \vdots & \multicolumn{1}{|c}{} & I_k \\
 & & \multicolumn{1}{c|}{} & 0 & 0 & \multicolumn{1}{|c}{} & \\\cline{1-3}\cline{6-8}
0 & \cdots & 0 & 0 & 1 & 0 & \cdots & 0
\end{array}\right]
\end{equation}
for any positive integer $k$, where $I_k$ is the $k \times k$ identity matrix.
\end{remark}

Using Lemma \ref{lem:transferMatrix}, we can now detail how to construct the boundary conditions.  The idea is simple:  We restrict the solution to lie in the nullspace of $(z-H)_\per$.  We do this by requiring that the solution is orthogonal to the orthogonal complement of the nullspace.

First, note that since $\{e_j\}_{j=1}^{N_b} \in \nul(z-H)_\per$, there are no constraints on the first $N_b$ entries of the solution.  Therefore, we only need to enforce the condition that the truncation of the solution is orthogonal to the eigenvectors of $R + \lambda S$.  This is easily done as follows.  Let $V \in \mathbb{C}^{N_R \times (N_R - N_d + N_b)}$ be a basis for the orthogonal complement of the eigenvectors of $R+\lambda S$.  Then we can solve for $g_j$ using the system,
\begin{equation}
\left[\begin{array}{p{.3cm}p{.3cm}p{.3cm}p{.3cm}p{.3cm}p{.3cm}}
& & & \multicolumn{1}{c|}{} & & \\
\multicolumn{4}{c|}{(z-H)_\text{defect}} &  \\\cline{5-5}
\multicolumn{4}{c|}{} & \multicolumn{1}{c|}{\mathcal{P}}\\\hline
 & & & \multicolumn{3}{|c}{V^\hermi}
\end{array}\right] g_j = e_j. \label{systemForgj}
\end{equation}

Finally, we prove that $g_j$ is the unique solution of this system.  

\begin{lemma}
Using our method for constructing the boundary conditions, \eqref{systemForgj} has a unique solution.  That unique solution is the first $N_R + N_b$ entries of $g_j$.
\end{lemma}
\begin{proof}
First, note that existence of a solution is clear: $g_j$ is a solution.  Next, let $v$ be a solution of \eqref{systemForgj}.  Then, $v$ can be extended to $\ell^2(\mathbb{N})$ and would satisfy Conditions 1' and 2'.  Therefore, by Lemma \ref{lem:uniqueness_of_g} and the discussion following it, $v$ must equal $g_j$.
\end{proof}

\section{Proofs for Appendix \ref{sec:method_in_general}} \label{sec:method_in_general_proofs}

\subsection{Proof of Lemma \ref{lem:nullspace_zH_per}}
\label{sec:proof_of_lemma_nullspace}

\begin{lemma}
\label{lem:dimOfNullspace}
Let $A$ be an invertible operator on $\ell^2(\mathbb{N})$.  In the standard basis for $\ell^2(\mathbb{N})$, consider $A$ as an infinite dimensional matrix.  Let $A_n$ be the operator $A$ with $n$ rows removed.  Then, $\dim(\nul(A_n)) = n$.
\end{lemma}
\begin{proof}
Let $R$ be the $n \times \infty$ matrix of the rows removed from $A$.  We first show that $R^\hermi$ and $A_n^\hermi$ span disjoint subspaces,
\begin{equation}
\ran(R^\hermi) \cap \bar{\ran(A_n^\hermi)} = 0. \label{intersectionOfRanges}
\end{equation}
We can prove this by contradiction: Suppose there is a nonzero vector $v$ in the intersection.  If $\{v_i\}_{i=1}^n$ are the columns of $R^\hermi$ and $\{u_k\}_{k=1}^\infty$ are the columns of $A_n^\hermi$, then 
\begin{equation}
v = \sum_{j=1}^{n} d_j v_j = \sum_{k=1}^\infty c_k u_k.
\end{equation}
This implies that $[\begin{array}{ccccccc} -d_1 & -d_2 & ... & -d_{n} & c_1 & c_2 & ... \end{array}]^\trans$ is in the nullspace of $A^\hermi$; this would then imply that $A$ is not invertible.  Thus, \eqref{intersectionOfRanges} must be true.

Combining \eqref{intersectionOfRanges} with the fact that the columns of $A^\hermi$ must span all of $\ell^2(\mathbb{N})$, we have that any vector in $\ell^2(\mathbb{N})$ can be written uniquely as a sum $v + u$ where $v \in \ran(R^\hermi)$ and $u \in \bar{\ran(A_n^\hermi)}$.

Next, we show that the orthogonal complement of $\bar{\ran(A_n^\hermi)}$ has the same dimension as $\ran(R^\hermi)$.  Let $P$ be the projection onto $\bar{\ran(A_n^\hermi)}^\perp$.  Then, for any vector $f = v+u \in \ell^2(\mathbb{N})$, the projection onto $\bar{\ran(A_n^\hermi)}^\perp$ is
\begin{equation}
Pf = Pv + Pu = Pv.
\end{equation}
Therefore, if $\{v_i\}_{i=1}^{n}$ is a basis for $\ran(R^\hermi)$, then
\begin{equation}
\bar{\ran(A_n^\hermi)}^\perp = \spn(\{Pv_i\}_{i=1}^{n}).
\end{equation}
Note that the set $\{Pv_i\}_{i=1}^{n}$ is linearly independent: If it was not then there would be a nonzero vector $v$ in $\ran(R^\hermi)$ such that $Pv = 0$; but this would imply that $v \in \bar{\ran(A_n^\hermi)}$, which is not possible.

Finally, using a well-known property of linear operators, we obtain the desired conclusion,
\begin{equation}
\dim(\nul(A_n)) = \dim(\ran(A_n^\hermi)^\perp) = n.
\end{equation}
\end{proof}

\begin{proof}[Proof of Lemma \ref{lem:nullspace_zH_per}]
By Lemma~\ref{lem:dimOfNullspace}, $\dim \left(\nul(z-H)_\per\right) = N_d$.  Let $\{v_i\}_{i=1}^{N_d}$ be an orthonormal basis of $\nul(z-H)_\per$.  By the commutativity of $(z-H)_\per$ and $S_p$, 	
\begin{align}
(z-H)_\per S_p v_i = 0, && i = 1, ..., N_d.
\end{align}
This implies that $\nul(z-H)_\per$ is an invariant subspace of $S_p$.  Let $V = [v_1, v_2, ..., v_{N_d}]$.  Then, the matrix $V^\hermi S_p V \in \mathbb{C}^{N_d \times N_d}$ has a Jordan normal form,
\begin{equation}
V^\hermi S_p V = MJM^{-1}
\end{equation}
Since $S_p$ is invariant on the range of $V$ (i.e., on the nullspace of $(z-H)_\per$), we can multiply on the left by $V$ (and the right by $M$) to obtain
\begin{equation}
S_p VM = VMJ. \label{partialJordanSp}
\end{equation}
This is essentially a partial Jordan normal form of $S_p$.  Since the columns of $V$ form a basis for $\nul(z-H)_\per$ and $M$ is nonsingular, \eqref{partialJordanSp} shows that there is a basis for $\nul(z-H)_\per$ that is made up of (generalized) eigenvectors of $S_p$.  In particular, such a basis is given by the columns of $VM$.  Lemma \ref{lem:eigenpairs_of_S_p} then implies that the nullspace of $(z-H)_\per$ is spanned by vectors that are either of form \eqref{evecForm} or that have a finite number of nonzero entries.  The vectors with a finite number of nonzero entries can be split into Type 1a and Type 1b by taking proper linear combinations.
\end{proof}

\subsection{Matrix pencils}

In order to prove Lemma \ref{lem:transferMatrix}, we will use results from the theory of matrix pencils.  We review the definition and relevant properties of matrix pencils here. For more information, see \cite{demmel1993generalized}.

\begin{definition}
A \textbf{\em matrix pencil} is a set of matrices of the form $\{A - \lambda B : \lambda \in \mathbb{C}\}$.  A matrix pencil is said to be \textbf{\em regular} if $\det(A - \lambda B)$ is not identically zero.  If $\det(A - \lambda B)$ is identically zero, then the matrix pencil is said to be \textbf{\em singular}.
\end{definition}

\begin{definition}
The \textbf{\em generalized eigenvalues} of the generalized eigenvalue problem $A - \lambda B = 0$ are the values of $\lambda$ for which $\det(A-\lambda B) = 0$.  If a vector $x$ satisfies $(A-\lambda_0 B)x = 0$ for some $\lambda_0$, then $x$ is said to be an eigenvector of $A-\lambda B$ corresponding to the generalized eigenvalue $\lambda_0$.
\end{definition}

Note that if a matrix pencil is singular, then the corresponding generalized eigenvalue problem has infinitely many generalized eigenvalues (any $\lambda \in \mathbb{C}$ is an eigenvalue).  A concrete example of this phenomenon is given by the generalized eigenvalue problem, 
\begin{equation}
\left(\left[\begin{array}{cc}
1 & 1 \\ 0 & 0
\end{array}\right]
+ \lambda \left[\begin{array}{cc}
1 & 0 \\ 0 & 0
\end{array}\right]\right) x = 0.
\end{equation}
This phenomenon is very different from that encountered in the standard eigenvalue problem, $(A - \lambda I)x = 0$, where there are always precisely $n$ eigenvalues if $A \in \mathbb{C}^{n \times n}$; additionally, the eigenspaces corresponding to distinct eigenvalues are not all independent.  This suggests that the theory of generalized eigenvalue problems is more complicated than that of the standard eigenvalue problem.  However, the generalized eigenproblems corresponding to regular matrix pencils maintain many nice properties that hold for standard eigenvalue problems.  In order to show this, we first introduce a concept integral to our proof of Lemma \ref{lem:transferMatrix} \dash the Kronecker canonical form.

\begin{theorem}[Kronecker canonical form]
\label{th:KroneckerCanonicalForm}
Consider a matrix pencil $A - \lambda B$.  It can be written into the canonical form,
\begin{equation}
A - \lambda B = P(\tilde{A} - \lambda \tilde{B}) Q^{-1},
\end{equation}
where $P$ and $Q$ are invertible square matrices, and $\tilde{A} - \lambda\tilde{B}$ is block diagonal.  Each block of $\tilde{A} - \lambda \tilde{B}$ is in one of four forms,
\begin{align*}
J_j(\alpha) &= \left[\begin{array}{ccccc}
\alpha - \lambda & 1 \\
 & \ddots & \ddots \\
 & & \alpha - \lambda & 1 \\
 & & & \alpha - \lambda
\end{array}\right], & N_j &= \left[\begin{array}{ccccc}
1 & - \lambda \\
& \ddots & \ddots \\
& & 1 & -\lambda \\
& & & 1
\end{array}\right], \\
L_j &= \left[\begin{array}{cccc}
-\lambda & 1 \\
 & \ddots & \ddots \\
 & & -\lambda & 1
\end{array}\right], & L_j^\trans &= \left[\begin{array}{cccc}
-\lambda \\
1 &  \ddots \\
& \ddots & -\lambda \\
& & 1
\end{array}\right],
\end{align*}
where $J_j(\alpha)$ and $N_j$ are $j \times j$ matrices, $L_j$ is $j \times (j+1)$ and $L^\trans_j$ is $(j+1) \times j$.  The value of $j$ may change from block to block.
\end{theorem}

The $J_j(\alpha)$ blocks correspond to finite generalized eigenvalue $\alpha$.  The $N_j$ blocks correspond to infinite generalized eigenvalues (i.e., correspond to generalized eigenvalue zero of the problem $\lambda A - B$).  $L_j$ and $L_j^\trans$ are called singular blocks.

Note that the Kronecker canonical form of a matrix pencil is very similar to the Jordan canonical form of a matrix.  However, in Jordan canonical form, there are only $J_j(\alpha)$ blocks.  The similarity between the two canonical forms becomes even closer when one restricts to regular matrix pencils.

\begin{lemma}
\label{lem:KroneckerFormOfRegularPencil}
If $A-\lambda B$ is a regular matrix pencil, then its Kronecker canonical form does not contain any singular blocks ($L_j$ or $L_j^\trans$).
\end{lemma}

From this lemma, it is easy to see that generalized eigenproblems corresponding to regular matrix pencils satisfy two properties familiar from standard eigenproblem theory:  (1) If $A, B \in \mathbb{C}^{n \times n}$, then there are precisely $n$ eigenvalues (including the infinite eigenvalues and counting algebraic multiplicity); and (2) The eigenspaces corresponding to distinct eigenvalues are all independent.

\subsection{Proof of Lemma \ref{lem:transferMatrix}}
\label{sec:proof_transferMatrix}

Before proving Lemma \ref{lem:transferMatrix}, we first show that the matrix pencil considered in that lemma is regular.

\begin{lemma}
\label{lem:pencilIsRegular}
The matrix pencil $R + \lambda S$ in Lemma \ref{lem:transferMatrix} is regular.
\end{lemma}
\begin{proof}
The idea is as follows:  If $R+\lambda S$ is singular, then its Kronecker canonical form contains an $L_j$ (or $L_j^\trans$) block; this fact can then be used to show that columns (or rows) of $(z-H)$ must be linearly dependent, which contradicts the assumption that $z \notin \sigma(H)$.  The details now follow.

First, we will only consider the case that the Kronecker canonical form
\begin{equation}
R + \lambda S = P(\tilde{R} + \lambda \tilde{S})Q^{-1}
\end{equation}
has an $L_j$ block; the $L_j^\trans$ case is similar.  Assume without loss of generality that the top-left block of $\tilde{R} + \lambda \tilde{S}$ is an $L_j$ block.  This implies that 
\begin{align}
(\tilde{R} + \lambda \tilde{S})\left[\begin{array}{c}
1\\ \lambda \\ \vdots \\ \lambda^j \\ \vec{0}
\end{array}\right] &= 0, & \text{for all $\lambda \in \mathbb{C}$.}
\end{align}
By considering the structure of $\tilde{R}$ and $\tilde{S}$ in the upper-left block, one can rewrite this as
\begin{equation}
\tilde{S} e_j \lambda^{j+1} + (\tilde{R} e_j + \tilde{S} e_{j-1}) \lambda^j + ... + (\tilde{R} e_2 + \tilde{S} e_1) \lambda + \tilde{R} e_1 = 0.
\end{equation}
Since this must be true for all $\lambda$, the coefficient in front of each $\lambda^k$ must be zero.  This then gives the set of equations,
\begin{equation}
\left[\begin{array}{ccccc}
\tilde{S} \\
\tilde{R} & \ddots \\
 & \ddots & \tilde{S} \\
 & & \tilde{R}
\end{array}\right] \left[\begin{array}{c}
e_j \\ e_{j-1} \\ \vdots \\ e_1
\end{array}\right] = 0.
\end{equation}
Finally, we multiply each block equation on the left by $P$ and multiply a $Q^{-1}Q$ in between the matrix-vector multiplications to obtain
\begin{equation}
\left[\begin{array}{ccccc}
S \\
R & \ddots \\
 & \ddots & S \\
 & & R
\end{array}\right] \left[\begin{array}{c}
Qe_j \\ Qe_{j-1} \\ \vdots \\ Qe_1
\end{array}\right] = 0.
\end{equation}
Since $Q$ has full rank, the vector is nonzero.  The columns of the matrix on the left are identical to the columns of $(z-H)$ as the matrix goes off to infinity (with the non-relevant zero-rows above and below removed).  This equation then implies that the columns of $(z-H)$ are linearly dependent.
\end{proof}

Using this lemma and the nice properties that follow from it (as mentioned below Lemma \ref{lem:KroneckerFormOfRegularPencil}), we can now prove Lemma \ref{lem:transferMatrix}.

\begin{proof}[Proof of Lemma \ref{lem:transferMatrix}]
As noted before the statement of the lemma, it is simple to check that restrictions of the Type 1b and Type 2 basis vectors satisfy $(R+\lambda S)\xi = 0$ with $|\lambda| < 1$.  Since $R + \lambda S$ is a regular pencil (Lemma \ref{lem:pencilIsRegular}), the Kronecker canonical form implies that eigenspaces corresponding to different eigenvalues are independent.  Now suppose that the set of restrictions of the Type 1b and Type 2 basis vectors of $\nul(z-H)_\per$ is linearly dependent.  This would imply that there is a dependent set would corresponds to the same eigenvalue.  However, this then implies that the extensions of these vectors are linearly dependent \dash contradicting the assumption that they were not.  Therefore, the restrictions of Type 1b and Type 2 vectors form a linearly independent set of $N_d - N_b$ eigenvectors of $R+\lambda S$.

Let $\xi$ be an eigenvector of $R+\lambda S$ with $|\lambda| < 1$.  Then, it is straightforward to show that the extension of $\xi$ is in $\nul(z-H)_\per$.  Therefore, the extension of $\xi$ must be a linear combination of the basis vectors for $\nul(z-H)_\per$.  Hence, $\xi$ is a linear combination of the restrictions of these basis vectors.  Therefore, there are precisely $N_d-N_b$ independent eigenvectors of $R + \lambda S$ with $|\lambda| < 1$.

Finally, it is straightforward to show that the extensions of any $N_d-N_b$ independent eigenvectors of $R+\lambda S$ (combined with $\{e_j\}_{j=1}^{N_b}$) gives a basis for $\nul(z-H)_\per$.
\end{proof}

\end{appendices}

\begingroup
\sloppy
\printbibliography
\endgroup

\end{document}